\documentclass[11pt,a4paper]{amsart}
\usepackage[foot]{amsaddr}
\usepackage{ifxetex}
\ifxetex
  \usepackage[no-math]{fontspec}
\else
\fi
\usepackage{amsmath}
\usepackage{amsfonts}
\usepackage{amssymb}
\usepackage{amsthm}
\usepackage{fullpage}
\usepackage{microtype}
\usepackage[libertine]{newtxmath}
\usepackage[tt=false]{libertine} 
\usepackage{caption}
\usepackage{bbm}
\usepackage{hyperref, color}
\hypersetup{colorlinks=true,citecolor=blue, linkcolor=blue, urlcolor=blue}
\usepackage[linesnumbered,boxed,ruled,vlined]{algorithm2e}
\usepackage{bm}
\usepackage{bbm}
\usepackage[numbers]{natbib}
\usepackage{xcolor}
\usepackage{enumerate} 
\usepackage{enumitem}
\usepackage{tabularx}
\usepackage{array}
\usepackage{cleveref}

\usepackage{todonotes}

\usepackage{xifthen}

\newcommand{\DTV}[2]{d_{\mathrm{TV}}\left({#1},{#2}\right)}
\newcommand{\E}[1]{\mathbb{E}\left[{#1}\right]}

\newcommand{\I}{\mathcal{I}}
\newcommand{\D}{\mathcal{D}}

\newcommand{\X}{\boldsymbol{X}}

\newcommand{\R}{\mathcal{R}}
\newcommand{\List}{\mathcal{L}}
\newcommand{\ReSample}{\textsf{Fix}}

\newcommand{\dist}{\mathrm{dist}}
\newcommand{\fsb}{w}
\newcommand{\Prob}{\Pr}

\newcommand{\Fss}{\mathcal{F}\text{ succeeds}}
\newcommand{\while}{\textbf{while}}

\newcommand{\abs}[1]{\left\vert#1\right\vert}
\def\*#1{\mathbf{#1}}
\def\+#1{\mathcal{#1}}
\def\-#1{\mathrm{#1}}

\def\Btime{\ensuremath{O(q^{\Delta^{\ell+1}})}}
\def\twoBtime{\ensuremath{O(q^{2\Delta^{\ell+1}})}}

\newtheorem{theorem}{Theorem}[section]
\newtheorem{observation}[theorem]{Observation}
\newtheorem{claim}[theorem]{Claim}
\newtheorem*{claim*}{Claim}
\newtheorem{condition}[theorem]{Condition}

\newtheorem{fact}[theorem]{Fact}
\newtheorem{property}[theorem]{Property}
\newtheorem{lemma}[theorem]{Lemma}
\newtheorem{proposition}[theorem]{Proposition}
\newtheorem{corollary}[theorem]{Corollary}
\theoremstyle{definition}

\newtheorem{definition}[theorem]{Definition}
\newtheorem{remark}[theorem]{Remark}
\newtheorem*{remark*}{Remark}

\crefname{theorem}{Theorem}{Theorems}
\crefname{observation}{Observation}{Observations}
\crefname{claim}{Claim}{Claims}
\crefname{condition}{Condition}{Conditions}
\crefname{algorithm}{Algorithm}{Algorithms}
\crefname{property}{Property}{Properties}
\crefname{example}{Example}{Examples}
\crefname{fact}{Fact}{Facts}
\crefname{lemma}{Lemma}{Lemmas}
\crefname{corollary}{Corollary}{Corollaries}
\crefname{definition}{Definition}{Definitions}
\crefname{remark}{Remark}{Remarks}
\crefname{proposition}{Proposition}{Propositions}
\crefname{equation}{equation}{equations}

\title{Perfect sampling from spatial mixing} 

\author{Weiming Feng}
\author{Heng Guo}
\author{Yitong Yin}

\address[Weiming Feng, Yitong Yin]{State Key Laboratory for Novel Software Technology, Nanjing University, 163 Xianlin Avenue, Nanjing, Jiangsu Province, China. \textnormal{E-mail: \url{fengwm@smail.nju.edu.cn, yinyt@nju.edu.cn}}}

\address[Heng Guo]{School of Informatics, University of Edinburgh, Informatics Forum, Edinburgh, EH8 9AB, United Kingdom. \textnormal{E-mail: \url{hguo@inf.ed.ac.uk}}}
\keywords{perfect sampling, Gibbs distribution, spatial mixing} 

\begin{document}

\begin{abstract}
  We introduce a new perfect sampling technique that can be applied to general Gibbs distributions and runs in linear time if the correlation decays faster than the neighborhood growth.
  In particular, in graphs with sub-exponential neighborhood growth like $\mathbb{Z}^d$, our algorithm achieves linear running time as long as Gibbs sampling is rapidly mixing.
  As concrete applications, we obtain the currently best perfect samplers for colorings and for monomer-dimer models in such graphs.

\end{abstract}

\maketitle

\section{Introduction}

Spin systems model nearest neighbor interactions of complex systems.
These models originated from statistical physics, and have found a wide range of applications in probability theory, machine learning, and theoretical computer science,
often under different names such as \emph{Markov random fields}.
Given an underlying graph $G=(V,E)$,
a \emph{configuration} $\sigma$ is an assignment from vertices to a finite set of spins, usually denoted by $[q]$.
The \emph{weight} of a configuration is specified by the $q$-dimensional vector $b_v$ assigned to each vertex $v \in V$  and the $q$-by-$q$ symmetric interaction matrix $A_e$ assigned to each edge $e \in E$, namely,
\begin{align}\label{eqn:spin-system}
  w(\sigma) = \prod_{v \in V}b_v(\sigma_v)\prod_{e = \{u,v\}\in E} A_e(\sigma_u, \sigma_v).
\end{align}
The equilibrium state of the system is described by the \emph{Gibbs distribution} $\mu$, 
where the probability of a configuration is proportional to its weight.

A central algorithmic problem related to spin systems is to sample from the Gibbs distribution.
A canonical Markov chain for sampling approximately from the Gibbs distribution is the \emph{Gibbs sampler} (a.k.a.~\emph{heat bath} or \emph{Glauber dynamics}).
One (conjectured) general criterion for the rapid mixing of this chain is the \emph{spatial mixing} property,
which roughly states that correlation decays rapidly in the system as distance increases.
It is widely believed that spatial mixing (in some form) implies the rapid mixing of the Gibbs sampler.
However, rigorous implications have only been established for special classes of graphs or systems,
such as for lattice graphs \cite{Mar99,DSVW04}, neighborhood amenable graphs \cite{GMP05}, or for monotone systems \cite{mossel2013exact}.

One main drawback for Gibbs samplers or Markov chains in general,
is that one needs to know the mixing time in advance to be able to implement them to make the error provably small.
The so-called \emph{perfect} samplers are thus more desirable,
which run in Las Vegas fashion and return exact samples upon halting.
There have been a number of techniques available to design perfect samplers,
such as Coupling From The Past (CFTP) \cite{propp1996exact} including the bounding chains \cite{huber2004perfect,bhandari2019perfect}, 
Randomness Recycler (RR) \cite{fill2000randomness}, and Partial Rejection Sampling (PRS) \cite{GJL19,FVY19}.
Nevertheless, none of these previous techniques addresses general spin systems or relates to the important spatial mixing properties of the system.

In this paper, we introduce a new technique to perfectly sample from Gibbs distributions of spin systems.
The correctness of our algorithm relies on only the conditional independence property of Gibbs distributions.
Moreover, the expected running time is linear in the size of the system, when the correlation decays more rapidly than the growth of the neighborhood.

\begin{theorem}[{informal}]  \label{thm:main-informal}
For any spin system with bounded maximum degree,\footnote{We remark that our algorithm remains in polynomial-time (but not in linear time) for graphs with unbounded degrees,
as long as the degree does not grow too quickly. 
This requirement on the degree comes from the cost of updating a block of certain radius,
similar to the cost of block dynamics,
and the exact upper bound needed varies from problem to problem.
See \Cref{theorem-general} and the discussion thereafter.
For the most part, we state our results for bounded degree cases to keep the statements clean.}
if strong spatial mixing holds with a rate faster than the neighborhood growth of the underlying graph,
then there exists a perfect sampler with running time $O(n)$ in expectation, where $n$ represents the number of vertices of the graph.
\end{theorem}

More details and undefined terms are explained in \Cref{section-results}.
Formal statements of our results are given in \Cref{theorem-sub-graph} for spin systems on sub-exponential neighborhood growth graphs, 
and in \Cref{theorem-general} for spin systems on general graphs.
Applications on list colorings and monomer-dimer models are given in \Cref{corollary-main-coloring} and \Cref{corollary-main-matching}.

Lattice graphs, such as $\mathbb{Z}^d$, are of special interests in statistical physics and combinatorics. 
These graphs have sub-exponential neighborhood growth,
which implies that temporal mixing is equivalent to spatial mixing on them~\cite{DSVW04}.
Therefore our sampler runs in linear time as long as the standard Glauber dynamics has $O(n\log n)$ mixing time.
This is a direct strengthening of aforementioned results~\cite{Mar99,DSVW04,GMP05} from approximate to perfect sampling, with an improved running time.

\begin{corollary}[{informal}]  \label{cor:main-informal}
For spin systems on graphs with sub-exponential neighborhood growth,
if the Gibbs sampler has $O(n\log n)$ mixing time, where $n$ is the number of vertices, or the system shows strong spatial mixing,
then there exists a perfect sampler with running time $O(n)$ in expectation.
\end{corollary}

It is worth noting that many traditional perfect sampling algorithms,
especially those rely on CFTP \cite{propp1996exact,huber2004perfect,bhandari2019perfect},
suffer from ``non-interruptibility''. 
That is, early termination of the algorithm induces a bias on the sample.
In contrast, our algorithm is interruptible in the following sense: 
conditioned on its termination at any particular step, the algorithm guarantees to return a correct sample.
Therefore, had the algorithm been running for too long, one can simply stop it and restart.
One can also run many independent copies  in parallel and output the earliest returned sample without biasing the output distribution. 

In addition, our algorithm can be used to solve the recently introduced dynamic sampling problem~\cite{FVY19,feng2019dynamicMCMC},
where the Gibbs distribution itself changes dynamically and the algorithm needs to efficiently maintain a sample from the current Gibbs distribution.
The detail of this part is given in Section~\ref{section-dynamic}.
Our perfect sampler also generalizes straightforwardly to Gibbs distributions with multi-body interactions 
(namely spin systems on hypergraphs / constraint satisfaction problems),
and similar efficiency can be achieved when some appropriate variant of spatial mixing holds.

Morally, we believe that efficient perfect sampling algorithms exist whenever efficient approximate samplers exist.
Our result verifies this belief for spin systems on lattice graphs.
However, in general the gap between approximate and perfect sampling persists (e.g.~for sampling proper colorings \cite{bhandari2019perfect}),
and designing efficient perfect sampling algorithms matching their approximate counterparts remain an interesting research direction.

\subsection{Algorithm overview}\label{sec:alg-overview}
In this section, we give a simplified single-site version of our algorithm to illustrate the main ideas.

Recall that the Gibbs sampler is a Markov chain on state space $\Omega=[q]^V$. 
The chain starts from an arbitrary initial configuration $\X\in[q]^V$. At each step, a vertex $u\in V$ is picked uniformly at random and $\X$ is updated as following:
\begin{itemize}
\item the spin $X_u$ is redrawn according to the marginal distribution $\mu_u^{X_{\Gamma(u)}}$ induced at vertex $u$ by Gibbs distribution $\mu$ conditioned on the current spins of the neighborhood $\Gamma(u)$,
\end{itemize}  
where $\Gamma(u)\triangleq\{v\in V\mid \{u,v\}\in E\}$ denotes the neighborhood of $u$ in graph $G=(V,E)$.
It is a basic fact that this chain converges to the desired Gibbs distribution $\mu$.

Our perfect sampler makes use of the same update rule.
For ease of exposition, here we consider only soft constraints (where all $A_e$ and $b_v$ are positive).
The single-site version of our perfect sampler is described in Algorithm~\ref{alg:perfect-sampler-ss}.

{
\let\oldnl\nl
\newcommand{\nonl}{\renewcommand{\nl}{\let\nl\oldnl}}

\begin{algorithm}[htbp]
\SetKwInOut{Input}{Input}
\SetKwComment{Comment}{\quad$\triangleright$\ }{}
\SetKwIF{withprob}{}{}{with probability}{do}{}{}{}
Start from an arbitrary initial configuration $\boldsymbol{X} \in [q]^V$ and $\R \gets V$\;
\While{$\R \neq \emptyset$}{
	pick a $u\in\R$ uniformly at random\; 
let $\mu_{\min}$ be the minimum value of $\mu_{u}^{\sigma}(X_u)$ over all $\sigma\in[q]^{\Gamma(u)}$ that $\sigma_{\R\cap\Gamma(u)}=X_{\R\cap\Gamma(u)}$\;
%
\withprob{\,\,${\mu_{\min}}/{\mu_{u}^{X_{\Gamma(u)}}(X_u)}$\,\,}{
\nonl\vspace{-10pt}\hspace{206.6pt}\texttt{$\triangleright$ Bayes filter}

update $\X$ by redrawing $X_u\sim\mu_{u}^{{X_{\Gamma(u)}}}$;
\hspace{44.3pt}\texttt{$\triangleright$ Gibbs sampler update}\\
	 $\mathcal{R} \gets\mathcal{R} \setminus \{u\}$\;
}
\Else{
	$\mathcal{R} \gets\mathcal{R} \cup \Gamma(u)$\;
}
}
\Return{$\X$}\;
\caption{Perfect Gibbs sampler (single-site version)}\label{alg:perfect-sampler-ss}
\end{algorithm}
}

The algorithm starts from an arbitrary initial configuration $\boldsymbol{X}\in [q]^V$, 
and gradually ``repairs'' $\boldsymbol{X}$ to a perfect sample of the Gibbs distribution $\mu$.
We maintain a set $\mathcal{R}\subseteq V$ of vertices that currently ``incorrect'', initially set as $\R=V$.
At each step, a random vertex $u$ is picked from $\R$, and we try to remove $u$ from $\R$ while maintaining the following invariant, 
known as the \emph{conditional Gibbs property}:
\begin{align}
\begin{array}{l}
X_{\overline{\mathcal{R}}}\text{ always follows the law } \mu_{\overline{\mathcal{R}}}^{X_\mathcal{R}},\\
\text{which is the marginal distribution induced by $\mu$ on $\overline{\mathcal{R}}\triangleq V\setminus\mathcal{R}$ conditioned on $X_\mathcal{R}$}.
\end{array}\label{eq:invariant-conditional-gibbs}
\end{align}
This property ensures that the configuration on $\overline{\+R}$ follows the correct distribution conditioned on the configuration on $\+R$.
In particular, when $\mathcal{R}=\emptyset$, $\overline{\+R}=V$, $\mu_{\overline{\mathcal{R}}}^{X_\mathcal{R}}=\mu$,
and the sample $\boldsymbol{X}$ follows precisely the distribution $\mu$. 
This is the goal of our algorithm: to reduce $\+R$ to the empty set.

Our first attempt is simply to update the spin $X_u$ according to its marginal distribution $\mu_u^{{X_{\Gamma(u)}}}$ 
as in the Gibbs sampler and then remove $u$ from $\R$. 
This gives the transition $(\X,\R)\to(\X',\R')$,
where $\R'=\R\setminus\{u\}$, and $\X=\X'$ except at vertex $u$, where  $X'_u\sim\mu_u^{{X_{\Gamma(u)}}}$.

Ideally, if we had $X_{\overline{\R}}\sim\mu_{\overline{\R}}^{X_{\R'}}$, since $\R'=\R\setminus\{u\}$,
after  the spin of $u$ being updated as $X'_u\sim\mu_u^{{X_{\Gamma(u)}}}$, the partial configuration $X_{\overline{\R}}$ would be extended to a {$X'_{\overline{\R'}}\sim \mu_{\overline{\R'}}^{X_{\R'}}\equiv\mu_{\overline{\R'}}^{X'_{\R'}}$} as $X_{\R'}=X'_{\R'}$, giving us the invariant~\eqref{eq:invariant-conditional-gibbs} on the new pair $(\X',\R')$.

However, the invariant~\eqref{eq:invariant-conditional-gibbs} on the original $(\X,\R)$ only guarantees $X_{\overline{\R}}\sim\mu_{\overline{\R}}^{X_{\R}}$ rather than $X_{\overline{\R}}\sim\mu_{\overline{\R}}^{X_{\R'}}$.
%
To remedy this, we construct a filter $\mathcal{F}=\mathcal{F}(u,\boldsymbol{X})$ that corrects  $\mu_{\overline{\R}}^{X_{\R}}$ to $\mu_{\overline{\R}}^{X_{\R'}}$.
We call $\+F$ the \emph{Bayes filter}.
Specifically, $\mathcal{F}$ 
is determined by a biased coin depending on only part of $\X$ so that
\begin{align}
\Pr[\text{\,$\mathcal{F}$ succeeds\,}]
&\propto
\frac{\mu_{\overline{\R}}^{X_{\R'}}(X_{\overline{\R}})}{\mu_{\overline{\R}}^{X_{\R}}(X_{\overline{\R}})}
=
\frac{\mu_u^{X_{\R'}}(X_u)}{\mu_{u}^{X_{\Gamma(u)}}(X_u)},\label{eq:filter-property-1}
\end{align}
where $\propto$ is taken over all ${X_{\overline{\R}}}$ and
the equality is due to Bayes' theorem:
\begin{align*}
\mu_{\overline{\R}}^{X_{\R}}(X_{\overline{\R}})
&=
\mu_{\overline{\R}}^{X_{\R'}\land X_u}(X_{\overline{\R}})
=
{\mu_{u}^{X_{\overline{\R}}\land X_{\R'}}(X_u)\cdot\mu_{\overline{\R}}^{X_{\R'}}(X_{\overline{\R}})}/{\mu_u^{X_{\R'}}(X_u)},
\end{align*}
together with the \emph{conditional independence property} (formally, \Cref{property-cond-ind}) of Gibbs distribution which guarantees that  
$\mu_{u}^{X_{\overline{\R}}\land X_{\R'}}(X_u)=\mu_{u}^{X_{\Gamma(u)}}(X_u)$.

We observe that despite that the exact value of the marginal probability $\mu_{u}^{{X_{\R'}}}(X_u)$ in~\eqref{eq:filter-property-1} is hard to compute, it does not depend on $X_{\overline{\R}}$. 
Therefore, $\mu_{u}^{X_{\R'}}(X_u)$ can be treated as a constant and \eqref{eq:filter-property-1} holds as long as 
$\Pr[\text{\,$\mathcal{F}$ succeeds\,}]\propto{1}/{\mu_{u}^{X_{\Gamma(u)}}(X_u)}$, 
which is satisfied by the Bayes filter in Algorithm~\ref{alg:perfect-sampler-ss}.

Now by~\eqref{eq:filter-property-1}, conditioned on the success of the filter $\mathcal{F}$, the new $(\X',\R')$ satisfies the invariant~\eqref{eq:invariant-conditional-gibbs}.
Meanwhile, since the filter only reveals the neighborhood spins $X_{\Gamma(u)}$, upon failure of $\mathcal{F}$, the invariant \eqref{eq:invariant-conditional-gibbs} remains to hold as long as the revealed sites $\Gamma(u)$ are included into $\R$ and $\X$ is unchanged.

This sampler is valid for general Gibbs distributions,
since the only property we require is conditional independence. 
However for efficiency purposes our general algorithm, \Cref{alg:perfect-sampler-gen}, 
uses \emph{block} updates rather than the single-site updates in \Cref{alg:perfect-sampler-ss}. 
Block updates guarantee that the filter $\mathcal{F}$ always has a positive chance to succeed for Gibbs distributions with hard constraints,
and a suitable block length (chosen according to spatial mixing) is the key to efficiency of our analysis.
Moreover, to make sure that marginal distributions are well-defined,
we restrict our attention to \emph{permissive} systems, which contain all soft constraint systems as well as all hard constraint systems of interest.
The details are given in Section~\ref{section-results} and~\ref{sec:algorithm}.

The algorithm is efficient as long as the size of $\R$ shrinks in expectation in every step. 
For the more general \Cref{alg:perfect-sampler-gen}, this is implied by correlation decay faster than neighborhood growth. 
The details are in Section~\ref{section-running-time}.
%

Let us remark that it is surprising to us that the Gibbs sampler, studied for decades as the go-to approximate sampling algorithm, 
can be turned into a perfect sampler by simply adding a filter that accesses only local information.

\subsection{Related work and open problems}
The conditional Gibbs property have been used implicitly in previous works such as partial rejection sampling~\cite{GJL19,GJ18,GJ19a,GJ19b}, dynamic sampling~\cite{FVY19}, and randomness recycler~\cite{fill2000randomness}. Furthermore, the invariant of the conditional Gibbs property ensures that the sampling algorithm is correct even when the input spin system is dynamically changing over time~\cite{FVY19}.

Before our work, all previous resampling algorithms~\cite{fill2000randomness,GJL19,GJ18,GJ19a,GJ19b,FVY19} fall into the paradigm of rejection sampling: a new sample is generated, usually from modifying the old sample, and (part of) the new sample is rejected independently with some probability determined by the new sample. In our algorithm, the filtration is executed \emph{before} the generation of the new sample, with a bias independent of the new sample.

%


While our algorithm works for spin systems in general,
it still requires a rather strong notion of spatial mixing to be efficient.
Weaker forms of spatial mixing are known to imply rapid mixing of the Gibbs sampler for some particular systems. See for example \cite{MS13}.
If by ``efficient'' we allow algorithms whose running times are high degree polynomials and may depend on the degree of the graph,
then indeed even weak spatial mixing corresponds to the optimal threshold for efficient samplability in anti-ferromagnetic 2-spin systems \cite{Wei06,LLY13,SST14,SS14,GSV16}.
It remains to be an interesting open problem whether spatial mixing implies the existence of efficient samplers in general,
and whether these samplers can be perfect.
Our algorithm handles the failure of the filter in a very pessimistic way --- all revealed variables are considered ``incorrect'' and are added to $\R$.
Perhaps better algorithms can be obtained by handling the failure case more efficiently.

\section{Our results}
\label{section-results}
\subsection{Model and definitions}
Let $G=(V, E)$ be an undirected graph, and $[q]=\{1,2,\ldots,q\}$ a finite domain of $q\ge2$ spins. 
%
%
An instance of $q$-state spin system is specified by a tuple $\I=(G, [q], \boldsymbol{b}, \boldsymbol{A})$, where $\boldsymbol{b} = (b_v)_{v \in V}$ assigns every vertex $v\in V$ a  vector $b_v \in \mathbb{R}_{\geq 0}^q$ and 
$\boldsymbol{A} = (A_e)_{e \in E}$ assigns every edge $e\in E$ a symmetric matrix $A_e \in \mathbb{R}_{\geq 0}^{q\times q}$. 
The \emph{Gibbs distribution} $\mu_{\I}$ over $[q]^V$ is defined as
\begin{align}
\label{eq-def-Gibbs}
\forall \sigma \in [q]^V: \quad \mu_{\I}(\sigma) \triangleq \frac{w_{\I}(\sigma)}{Z_{\I}} = \frac{1}{Z_{\I}}\prod_{v \in V}b_v(\sigma_v)\prod_{e = \{u,v\}\in E} A_e(\sigma_u, \sigma_v),	
\end{align}
where $w_{\I}(\sigma)$ is the \emph{weight} defined in~\eqref{eqn:spin-system} and $Z_{\I} \triangleq \sum_{\sigma \in [q]^V} w_{\I}(\sigma)$ is the \emph{partition function}.

We restrict our attention to the so-called ``{permissive}'' spin systems, where the marginal distributions are always well-defined.  
%
 Let $\I=(G, [q], \boldsymbol{b}, \boldsymbol{A})$ be an instance of spin system. 
A configuration on $V$ is called \emph{feasible} if its weight is positive, and a partial configuration  is \emph{feasible} if it can be extended to a feasible configuration.
For any (possibly empty) subset $\Lambda \subseteq V$ and any (not necessarily feasible) partial configuration $\sigma \in [q]^\Lambda$, we use $\fsb_{\I}^{\sigma}\left(\tau\right)$ to denote the weight of $\tau\in[q]^{V \setminus \Lambda}$ conditional on $\sigma$:
\begin{align}
\label{eq-def-proper}
\fsb_{\I}^{\sigma}\left(\tau\right) 
  = \prod_{v \in V \setminus \Lambda}b_v(\tau_v)
  \prod_{\substack{e=\{u,v\} \in E \\ u, v \in V \setminus \Lambda}}A_e(\tau_u, \tau_v)
  \prod_{\substack{e=\{u,v\} \in E \\ u \in \Lambda, v\in V  \setminus \Lambda}}A_e(\sigma_u, \tau_v).
\end{align}
Define the partition function  $Z_{\I}^{\sigma}$ conditional on $\sigma$ as
$Z_{\I}^{\sigma} \triangleq \sum_{\tau \in [q]^{V \setminus \Lambda}}\fsb_{\I}^{\sigma}\left( \tau\right)$.	

\begin{definition}[permissive]\label{definition-locally-admissible}
A spin system $\I=(G, [q], \boldsymbol{b}, \boldsymbol{A})$, where $G=(V,E)$, is called \emph{permissive} if $Z_{\I}^{\sigma} > 0$ for any partial configuration $\sigma \in [q]^{\Lambda}$ specified on any subset $\Lambda \subseteq V$. 
\end{definition}

%
%
%
Permissive systems are very common, 
including, for examples, 
uniform proper $q$-coloring when $q \geq \Delta + 1$, where $\Delta$ is the maximum degree,
and spin systems with soft constraints, e.g.~the Ising model, or with a ``permissive'' state that is compatible with all other states, e.g.~the hardcore model.

For permissive systems, a feasible configuration is always easy to construct by greedy algorithm.
More importantly, with permissiveness, marginal probabilities are always well defined, which is crucial for Gibbs sampler and spatial mixing property.

Formally, we use $\mu^\sigma_{\I}$ to denote the conditional distribution over $[q]^{V \setminus \Lambda}$ given $\sigma \in [q]^{\Lambda}$, that is,
\begin{align}
\label{eq-def-conditional-Gibbs}
\forall \tau \in [q]^{V \setminus \Lambda},\quad 	\mu^\sigma_{\I}(\tau) \triangleq \frac{\fsb_{\I}^{\sigma}\left( \tau \right)}{Z_{\I}^{\sigma}}.
\end{align}
And for any $v \in V \setminus \Lambda$, we use $\mu_{v,\I}^\sigma$ to denote the marginal distribution at $v$ projected from $\mu^\sigma_{\I}$. 


For any $u,v \in V$, we use $\dist_G(u,v)$ to denote the shortest-path distance between $v$ and $u$ in $G$. 

\begin{definition}[strong spatial mixing~\cite{weitz2006counting,weitz2004mixing}]
\label{definition-standard-SSM}
Let $\delta:\mathbb{N}\to\mathbb{R}^+$.
A class $\mathfrak{I}$ of permissive spin systems is said to exhibit \emph{strong spatial mixing} with rate $\delta(\cdot)$ if for every instance $\I=(G, [q], \boldsymbol{b}, \boldsymbol{A})\in\mathfrak{I}$, where $G=(V,E)$, for every $v \in V$, $\Lambda \subseteq V$, and any two partial configurations $\sigma, \tau \in [q]^\Lambda$, 
\begin{align}
\label{eq:condition-standard-SSM}
\DTV{\mu_{v, \I}^\sigma}{\mu_{v, \I}^\tau} \le \delta(\ell),
\end{align}
where $\ell = \min\{\dist_G(v, u) \mid u \in \Lambda,\ \sigma_u\neq\tau_u\}$, and $\DTV{\mu_{v, \I}^\sigma}{\mu_{v, \I}^\tau}\triangleq \frac{1}{2}\sum_{a \in [q]}\left\vert\mu_{v, \I}^\sigma(a)-\mu_{v, \I}^\tau(a) \right\vert$
denotes the total variation distance between $\mu_{v, \I}^\sigma$ and $\mu_{v, \I}^\tau$.
In particular, $\mathfrak{I}$ exhibits \emph{strong spatial mixing with exponential decay} if \eqref{eq:condition-standard-SSM} is satisfied for $\delta(\ell)= \alpha\exp(-\beta\ell)$ for some constants $\alpha,\beta>0$.
\end{definition}



Our first result holds for spin systems on graphs with bounded neighborhood growth.

\begin{definition}[sub-exponential neighborhood growth]\label{definition-sub-exp-graph}
A class $\mathfrak{G}$ of graphs is said to have \emph{sub-exponential neighborhood growth} if there is a function $s:\mathbb{N}\to\mathbb{N}$ such that $s(\ell)=\exp(o(\ell))$ and for every graph $G=(V,E)\in \mathfrak{G}$,
\begin{align*}
\forall v \in V, \forall \ell \geq 0, \quad  |S_\ell(v)| \leq s(\ell),
\end{align*}
where $S_{\ell}(v)\triangleq \{u \in V \mid \dist_G(v, u) = \ell\}$ denotes the {sphere} of radius $\ell$ centered at $v$ in $G$.
\end{definition}
Note that graphs with sub-exponential neighborhood growth necessarily have bounded maximum degree because we can set $\ell=1$ and get $s(1)=O(1)$. 

%

\subsection{Main results}
Our first result shows that for spin systems on graphs with sub-exponential neighborhood growth, strong spatial mixing implies the existence of linear-time perfect sampler.

\begin{theorem}[main theorem: bounded-growth graphs]
\label{theorem-sub-graph}
Let $q>1$ be a finite integer and $\mathfrak{I}$ a class of permissive $q$-state spin systems on graphs with sub-exponential neighborhood growth.
If $\mathfrak{I}$ exhibits strong spatial mixing with exponential decay, then there exists an algorithm such that given any instance $\I=(G,[q],\boldsymbol{b},\boldsymbol{A})\in \mathfrak{I}$, the algorithm outputs a perfect sample from $\mu_{\I}$ within $O\left( n \right)$ time in expectation, where $n$ is the number of vertices in $G$. 
\end{theorem}
\noindent
The factor in $O(\cdot)$ is the cost for a block update in (block) Gibbs sampler with $\ell_0$-radius blocks, where $\ell_0=O(1)$ is determined by both rates of correlation decay and neighborhood growth (by~\eqref{eq-set-ell}).

%
It is already known that for spin systems on sub-exponential neighborhood growth graphs,  the strong spatial mixing with exponential decay implies $O(n \log n)$ mixing time for block Gibbs sampler \cite{DSVW04}, which only generates approximate samples. We give a perfect sampler with $O(n)$ expected running time under the same condition.
The notion of sub-exponential neighborhood growth is related to, but should be distinguished from neighborhood-amenability (see e.g.~\cite{goldberg2005strong}), 
which says that in an infinite graph, for any constant real $c > 0$, there is an $\ell$ such that $\frac{|S_{\ell+1}(v)|}{|B_{\ell}(v)|} \leq c$ holds everywhere.
%

Our main result on general graphs assumes the following strong spatial mixing condition.

\begin{condition}
\label{condition-lower-bound}
%
%
%
Let $\I=(G, [q], \boldsymbol{b}, \boldsymbol{A})$ be a permissive spin system where $G=(V,E)$.
There is an integer $\ell=\ell(q)\ge 2$ such that the following holds:
for every $v \in V$, $\Lambda \subseteq V$, 
%
%
for any two partial configurations  $\sigma, \tau \in [q]^\Lambda$ satisfying $\min\left\{\dist_G(v, u) \mid u \in \Lambda,\ \sigma(u)\neq\tau(u)\right\} = \ell$,
\begin{align}
\label{eq:condition-stronger-SSM-general}
\DTV{\mu_{v, \I}^\sigma}{\mu_{v, \I}^\tau} 
\leq \frac{\gamma}{ 5|S_{\ell}(v)|},
\end{align}
where 
$S_{\ell}(v)$ denotes the sphere of radius $\ell$ centered at $v$ in $G$, and 
\begin{align}
\label{eq:condition-stronger-lower}
\gamma=\gamma(v,\Lambda)
\triangleq
\min\left\{\mu_{v,\I}^{\rho}(a)\mid \rho \in [q]^{\Lambda}, 
 a\in[q]\text{ that }\mu_{v,\I}^{\rho}(a)>0\right\}
\end{align}
denotes the lower bound of positive marginal probabilities at $v$.
\end{condition}



The above condition basically says that the spin systems exhibit strong spatial mixing with a decay rate faster than that of neighborhood growth, given that the marginal probabilities are appropriately lower bounded (which holds with $\gamma=\Omega(1)$ when entries of $\boldsymbol{A}$ and $\boldsymbol{b}$ are of finite precision and the maximum degree $\Delta$ is finitely bounded).
%
%
Our result on general graphs is stated as the following theorem. 


\begin{theorem}[main theorem: general graphs]
\label{theorem-general}
Let $\mathfrak{I}$ be a class of permissive spin systems satisfying Condition~\ref{condition-lower-bound}.
There exists an algorithm which given any instance $\I=(G,[q],\boldsymbol{b},\boldsymbol{A})\in\mathfrak{I}$, outputs a perfect sample from $\mu_{\I}$ within $ n\cdot q^{O\left(\Delta^\ell \right)} $ time in expectation, 
where $n$ is the number of vertices in $G$, $\Delta$ is the maximum degree of $G$, and $\ell=\ell(q)$ is determined by Condition~\ref{condition-lower-bound}. 
\end{theorem}

The $q^{O(\Delta^\ell)}$ factor in the time cost is contributed by the block Gibbs sampler update on $\ell$-radius blocks.
This extra cost could remain polynomial if $q=\omega(1)$, but the upper bound on $q$ for that will vary from problem to problem.
See e.g.~the discussion after \Cref{corollary-main-coloring}.
The conditions of both Theorem~\ref{theorem-sub-graph}
and Theorem~\ref{theorem-general} 
are special cases of a more technical condition (\Cref{condition-SSM-ratio}),
formally stated in Section~\ref{section-running-time}.

\subsection{Applications on specific systems}
Our results can be applied on various spin systems.
We consider two important examples:

\begin{itemize}
\item{\bf Uniform list coloring:}
A list coloring instance is specified by $\I=(G, [q],\List)$, where $\List \triangleq \{L_v\subseteq [q] \mid v\in V \}$ assigns each vertex $v\in V$ a list of colors $L_v\subseteq [q]$.
A proper list coloring of instance $\I$ is a $\sigma \in [q]^V$ that $\sigma_v \in L_v$ for all $v \in V$ and $\sigma_u \neq \sigma_v$ for all $\{u,v\} \in E$.
Let $\mu_{\I}$ denote the uniform distribution over all proper list colorings of $\I$.

\item{\bf The monomer-dimer model:} 
The monomer-dimer model defines a distribution over graph matchings.
An instance is specified by $\I = (G, \lambda)$, where $G = (V, E)$ is a graph and $\lambda> 0$. 
Each matching $M\subseteq E$ in $G$ is assigned a weight $w_{\I}(M) = \lambda^{|M|}$. 
Let $\mu_{\I}$ be the distribution over all matchings in $G$ such that $\mu_{\I}(M) \propto w_{\I}(M)$.
\end{itemize}
We use $\deg_G(v)$ to denote the degree of $v$ in $G$ and $\Delta=\max_{v\in V}\deg_G(v)$ the maximum degree.

First, for the list coloring problem, we define the following two conditions for instance $\I=(G,[q],\+L)$.
Let $\deg_G(v)$ denote the degree of $v \in V$ in graph $G$, and $\Delta\triangleq\max_{v\in V}\deg_G(v)$ the maximum degree.
%


\begin{condition}
\label{condition-subexp}
For every $v \in V$, 
$|L(v)|	\geq \alpha \deg_G(v) + \beta$,
where either one of the followings holds:
\begin{itemize}
\item $\alpha = 2$ and $\beta = 0$;
\item $G$ is triangle-free,  $\alpha > \alpha^*$ where $\alpha^* = 1.763\cdots$ is the positive root of $x^x = \mathrm{e}$,
and $\beta \geq \frac{\sqrt{2}}{\sqrt{2}-1}$ satisfies $(1-1/\beta) \alpha \mathrm{e}^{\frac{1}{\alpha}(1-1/\beta)} > 1$.
\end{itemize}

\end{condition}

\begin{condition}
\label{condition-list-coloring}
For every $v \in V$, 
$|L(v)| \geq \Delta^2 - \Delta + 2$.
\end{condition}

Our perfect sampler for list coloring runs in linear time in either above condition.

\begin{theorem}
\label{corollary-main-coloring}
Let $\mathfrak{L}$ be a class of list coloring instances with at most $q$ colors for a finite $q>0$. 
If either of the two followings holds for all instances $\I=(G,[q],\List)\in\mathfrak{L}$:
\begin{itemize}
\item \Cref{condition-subexp}  
and $G$ has sub-exponential neighborhood growth; or
\item \Cref{condition-list-coloring},
\end{itemize}
then there exists a perfect sampler for $\mu_{\I}$ that runs in expected $O\left( n \right)$ time, where $n$ is the number of vertices.
\end{theorem}
\noindent
The constant factor in $O(\cdot)$, can be determined in the same way as in Theorem~\ref{theorem-sub-graph} in the case of \Cref{condition-subexp}, 
but is much higher in the case of \Cref{condition-list-coloring} due to the arbitrary neighborhood growth.
Nevertheless, even in this costly case, $\ell=O(q^2\log q)$ and the overall overhead can by upper bounded by a rough estimate $\exp(\exp(\mathrm{poly}(q)))$.


Sampling proper $q$-colorings (where the lists are identical for all vertices) has been extensively studied, especially using Markov chains.
Approximate sampling 
received considerable attention \cite{jerrum1995very,dyer2003randomly,hayes2003non, molloy2004glauber,dyer2013randomly}.
The current best result~\cite{vigoda2000improved,chen2019improved} is the $O(n\log n)$ mixing time of the flip chain if $q \geq (\frac{11}{6}-\epsilon_0)\Delta$ for some constant $\epsilon_0 > 0$.
For perfect sampling $q$-colorings, 
Huber introduced a bounding chain \cite{huber2004perfect} based on CFTP~\cite{wilson1996generating},
which terminates within $O(n\log n)$ steps in expectation if $q\ge\Delta^2+2\Delta$.
In a very recent breakthrough, Bhandari and Chakraborty~\cite{bhandari2019perfect} introduced a novel bounding chain that has expected running time $O(n \log^2n)$ in a substantially broader regime $q > 3\Delta$.
Another way to obtain perfect samplers is to use standard reductions between counting and sampling~\cite{jerrum1986random}.
Using this technique, any FPTAS for the number of colorings can be turned into a polynomial-time {perfect} sampler.
(It is important that the approximate counting algorithm is deterministic,
or at least with errors that can be detected.) 
Currently, the FPTAS with the best regime for general graphs is due to Liu et al.\ \cite{liu2019deterministic}, which requires $q\ge 2\Delta$,
and it runs in time $n^{\mathsf{EXP}(\Delta)}$. 

Comparing to the results above, our algorithm draws perfect samples and achieves the $O(n)$ expected running time.
In case of sub-exponential growth graphs such as $\mathbb{Z}^d$, it improves the result of \cite{bhandari2019perfect} by requiring only $q > \alpha\Delta+O(1)$, where $\alpha>\alpha^*=1.763\cdots$.

\medskip
Next we consider the monomer-dimer model.
It was proved in~\cite{bayati2007simple,song2019counting} that instances $\I=(G,\lambda)$ on graphs $G$ with maximum degree $\Delta$ exhibits strong spatial mixing with rate $(1-\Omega(1/\sqrt{1+\lambda\Delta}))^{-\ell}$. 
Applying our algorithm yields the following perfect sampling result.

\begin{theorem}
\label{corollary-main-matching}
Let $\mathfrak{M}$ be a class of monomer-dimer instances $\I=(G,\lambda)$ on graphs $G$ with sub-exponential neighborhood growth and constant $\lambda$.
There exists an algorithm which given any instance $\I=(G,\lambda)\in\mathfrak{M}$, outputs a perfect sample from $\mu_{\I}$ within expected $O(n)$ time. 
\end{theorem}

Previously, Markov chains were the most successful techniques for sampling weighted matchings. 
The Jerrum-Sinclair chain~\cite{jerrum1989approximating} on a monomer-dimer model $\I=(G,\lambda)$, 
 generates {approximate samples} from $\mu_{\I}$ within $\widetilde{O}(n^2m)$ steps, where $m = |E|$. 
This chain also mixes in $O(n\log^2 n)$ time for finite subgraphs of the 2D lattice $\mathbb{Z}^2$ \cite{vdBB00}.

It is difficult to convert the Jerrum-Sinclair chain to perfect samplers.
Before our work, the only perfect sampler for the monomer-dimer model we are aware of 
is the one obtained via standard reductions from sampling to counting~\cite{jerrum1986random} (similar to the case of colorings),
together with deterministic approximate counting algorithms~\cite{bayati2007simple}.
This is a perfect sampler with running time $n^{\mathrm{Poly}(\Delta,\lambda)}$.

Our algorithm is the first linear-time perfect sampler for the monomer-dimer model on graphs with sub-exponential neighborhood growth,
such as finite subgraphs of lattices $\mathbb{Z}^d$ for any constant $d$ and constant weight $\lambda$.

\section{The Algorithm}
\label{sec:algorithm}
%
%

We now describe our general perfect Gibbs sampler.
It generalizes the single-site version  (Algorithm~\ref{alg:perfect-sampler-ss}) by allowing block updates.
This generalization allows us to bypass some pathological situations,
and to greatly improve the efficiency of the algorithm.
The pseudocode is given in Algorithm~\ref{alg:perfect-sampler-gen}.
%

Let $\I=(G, [q], \boldsymbol{b}, \boldsymbol{A})$ be a permissive spin system instance and $G=(V,E)$.
For any $u\in V$ and integer $\ell\ge 0$, recall that $B_{\ell}(u) \triangleq \left\{v \in V \mid \dist_{G}(u,v) \le \ell\right\}$ denotes the $\ell$-ball centered at $u$ in $G$.
And for any $B\subseteq V$, we use $\partial B\triangleq\{v\in V\setminus B\mid \{u,v\}\in E\}$ to denote the vertex boundary of $B$ in $G$.

{
\let\oldnl\nl
\newcommand{\nonl}{\renewcommand{\nl}{\let\nl\oldnl}}

\begin{algorithm}[htbp]
\SetKwInOut{Input}{Input}
\SetKwInOut{Input}{Parameter}
\SetKwComment{Comment}{\quad$\triangleright$\ }{}
\SetKwIF{withprob}{}{}{with probability}{do}{}{}{}
\Input{an integer $\ell \geq 0$;}
Start from an arbitrary feasible configuration $\boldsymbol{X} \in [q]^V$, i.e.~$w_{\I}(\boldsymbol{X})>0$\; 
$\R \gets V$\;
\While{$\R \neq \emptyset$}{
	pick a $u\in\R$ uniformly at random and let $B\gets (B_{\ell}(u)\setminus \R)\cup\{u\}$\label{line-pick}\; 
        let $\mu_{\min}$ be the minimum value of $\mu_{u}^{\sigma}(X_u)$ over all $\sigma\in[q]^{\partial B}$ that $\sigma_{\R\cap\partial B}=X_{\R\cap\partial B}$\label{line-F}\;
\withprob{$\frac{\mu_{\min}}{\mu_{u}^{X_{\partial B}}(X_u)}$}{
\nonl\vspace{-10pt}\hspace{204.5pt}\texttt{$\triangleright$ Bayes filter}

update $\X$ by redrawing $X_B\sim\mu_{B}^{X_{\partial B}}$;\label{line-sample}
\hspace{44.3pt}\texttt{$\triangleright$ block Gibbs sampler update}\\
	 $\mathcal{R} \gets\mathcal{R} \setminus \{u\}$\;
}
\Else{
	$\mathcal{R} \gets\mathcal{R} \cup \partial B$\;
}
}
\Return{$\X$}\;
\caption{Perfect Gibbs sampler (general version)}\label{alg:perfect-sampler-gen}
\end{algorithm}
}

The algorithm is parameterized by an integer $\ell\ge 1$, which is set rigorously in~\Cref{section-running-time}.



The initial $\X \in [q]^V$ is an arbitrary feasible configuration, which is easy to construct by greedy algorithm since $\I$ is permissive (\Cref{definition-locally-admissible}).
After each iteration of the \textbf{while} loop, 
either $\boldsymbol{X}$ is unchanged or $X_B$ is redrawn from $\mu_B^{X_{\partial B}}$, which is the marginal distribution of $\mu=\mu_{\I}$,  conditioned on the current $\X_B$. 
Thus, we have the following observation.
\begin{observation}
\label{observation-positive}
In \Cref{alg:perfect-sampler-gen}, the configuration $\boldsymbol{X}$ is always feasible,~i.e. $w_{\I}(\boldsymbol{X})>0$.
\end{observation}
\noindent  
The observation implies that $\mu_{u}^{X_{\partial B}}(X_u) > 0$ all along. 
%
The {Bayes filter} in Line~\ref{line-F} is always well-defined.
If the filer succeeds, $X_B$ is resampled according to the correct marginal distribution $\mu_B^{X_{\partial B}}$ and $u$ is removed from $\R$ (that is, $u$ has been successfully ``fixed'');
otherwise, $\X$ is unchanged and  $\R$ is enlarged by $\partial B$ (because variables in $\partial B$ are revealed and no longer random).
%
%



The key to the correctness of \Cref{alg:perfect-sampler-gen} is the conditional Gibbs property~\eqref{eq:invariant-conditional-gibbs}: the law of $\X$ over $\overline{R}\triangleq V \setminus \R$ is always the conditional distribution $\mu^{X_\R}=\mu^{X_\R}_{\I}$.
%
%
By similar argument as in Section~\ref{sec:alg-overview}, just redrawing $X_B$ from $\mu_{B}^{X_{\partial B}}$ will introduce a bias $\propto\mu_{u}^{X_{\partial B}}(X_u)$ to the sample $\X$, relative to its target distribution $\mu^{X_{\R\setminus\{u\}}}=\mu_{\I}^{X_{\R\setminus\{u\}}}$.
In the algorithm, we use the Bayes filter that succeeds with probability $\propto1/\mu_{u}^{X_{\partial B}}(X_u)$ to cancel this bias, 
with the risk of enlarging $\R$ by $\partial B$ upon failure.
Balancing the success probability and the size of $\partial B$ is the key to getting an efficient algorithm,
and this depends on choosing an appropriate $\ell$ according to the spatial mixing rate.

The correctness and efficiency of the algorithm are then analyzed in next two sections.

\section{Correctness of the perfect sampling}
\label{section-correctness}
In this section, we prove the correctness of \Cref{alg:perfect-sampler-gen},
%
stated by the following theorem.

\begin{theorem}[correctness theorem]
\label{theorem-correctness}
Given any permissive spin system $\I=(G,[q],\boldsymbol{b},\boldsymbol{A})$,
\Cref{alg:perfect-sampler-gen} with any parameter $\ell\ge 1$ terminates with probability 1, 
and outputs $\X$ that follows the law of $\mu_{\I}$.	
\end{theorem}

The theorem is implied by two key properties of the Gibbs distribution $\mu_\I$.

\subsection{Key properties of Gibbs distributions}
Note that the $\mu_{\min}$ in \Cref{alg:perfect-sampler-gen} is determined by  the set $\+R$, the vertex $u \in \+R$, and the partial feasible configuration $X_{\+R}$. Formally, fixing the parameter $\ell \geq 0$ in \Cref{alg:perfect-sampler-gen}, 
\begin{align*}
\mu_{\min}(\+R, u, X_{\R}) \triangleq \min\left\{ \mu^{\sigma}_{u,\I}(X_u) \mid \sigma \in [q]^{\partial B} \text{ s.t. } \sigma_{\+R \cap \partial B } = X_{\+R \cap \partial B}, \text{ where } B \triangleq (B_{\ell}(u) \setminus \+R) \cup \{u\} \right\}.
\end{align*}


\begin{property}[positive lower bound of $\mu_{\min}$]
\label{property-lower-bound}
The lower bound $\gamma_{\I}$ of $\mu_{\min}$ is positive:
\begin{align}
\label{eq-lower-bound-suss}
\gamma_{\I} \triangleq \min\left\{ \mu_{\min}(\R, u, X_{\R}) \mid \R \subseteq V, u \in \R, X_{\R} \in [q]^{\R} \text{ s.t. } X_{\R} \text{ is feasible} \right\} > 0.	
\end{align}
\end{property}

To state the next property, we introduce some notations:
For any $\Lambda \subseteq V$, $\sigma \in [q]^\Lambda$ and $S \subseteq V \setminus \Lambda$, we use $\mu^{\sigma}_{S,\I}(\cdot)$ to denote the marginal distribution on $S$ projected from $\mu^{\sigma}_{\I}$. For any disjoint sets $\Lambda,\Lambda' \subseteq V$, $\sigma \in [q]^{\Lambda}$ and $\sigma' \in [q]^{\Lambda'}$, we use $\sigma \uplus \sigma'$ to denote the configuration on $\Lambda \uplus \Lambda'$ that is consistent with $\sigma$ on $\Lambda$ and consistent with $\sigma'$ on $\Lambda'$.
\begin{property}[conditional independence]
\label{property-cond-ind}
Suppose $A,B,C\subset V$ are three disjoint non-empty subsets such that the removal of $C$ disconnects $A$ and $B$ in $G$.
For any $\sigma_A \in [q]^A, \sigma_B \in [q]^B$ and $\sigma_C \in [q]^C$, 
\begin{align*}
\mu^{\sigma_A \uplus \sigma_C}_{B,\I}(\sigma_B) = \mu^{\sigma_C}_{B, \I}(\sigma_B). 
\end{align*}
\end{property}

Theorem~\ref{theorem-correctness} is proved using only these two properties, 
thus \Cref{alg:perfect-sampler-gen} is correct for general permissive Gibbs distributions  satisfying these two properties.

In particular, we verify that all permissive spin systems satisfy these two properties.
First, the conditional independence (\Cref{property-cond-ind}) holds generally for Gibbs distributions~\cite{mezard2009information}. 
Next, for the positive lower bound of $\mu_{\min}$ (\Cref{property-lower-bound}): for spin systems with soft constraints, clearly \Cref{property-lower-bound} holds for all $\ell \geq 0$; 
and for general permissive spin systems $\I$, we need to verify that  \Cref{property-lower-bound} holds if $\ell \geq 1$. 
Fix a tuple $(\R,u,X_{\R})$ in~\eqref{eq-lower-bound-suss}.
The following fact follows from the definition of set $B$.
\begin{fact}
\label{fact-0}
$\partial B \subseteq S_{\ell+1}(u) \cup \+R$, where $B = (B_{\ell}(u) \setminus \+R) \cup \{u\}$.	
\end{fact}
\noindent
The fact implies $\partial B \setminus \+R \subseteq S_{\ell +1}(u)$.  
Since $\ell \geq 1$, $u$ is not adjacent to any vertex in $\partial B \setminus \+R$. 
Since $\I$ is permissive and $X_{\R}$ is feasible, $\mu_{u,\I}^{\sigma}(X_u) > 0$ for all $\sigma \in [q]^{ \partial B }$ such that $\sigma_{\+R \cap \partial B } = X_{\+R \cap \partial B}$.
This implies $\mu_{\min}(\R,u,X_{\R})$ is positive, thus the \Cref{property-lower-bound} holds.

We then prove \Cref{theorem-correctness} assuming only \Cref{property-lower-bound} and  \Cref{property-cond-ind}.
More specifically, termination of the algorithm is guaranteed by \Cref{property-lower-bound}, and correctness of the output upon termination is guaranteed by \Cref{property-cond-ind}.

\subsection{Termination of the algorithm}
Denote by $T$ the number of iterations of the \textbf{while} loop in \Cref{alg:perfect-sampler-gen}.
To prove that the algorithm terminates with probability 1,  we show that $T$ is stochastically dominated by a geometric distribution.  
We use $\+F$ to denote the Bayes filter in \Cref{alg:perfect-sampler-gen}. Then,
\begin{align*}
\Prob[\Fss] = \frac{\mu_{\min}(\+R, u, X_{\R}) }{\mu_{u,\I}^{X_{\partial B}}(X_u)} \geq \mu_{\min}(\+R, u, X_{\R}).	
\end{align*}
%

If $\+F$ succeeds for $n = |V|$ consecutive iterations of the \textbf{while} loop, then the set $\R$ must become empty and the algorithm terminates.
By \Cref{property-lower-bound},
we have
\begin{align}
\label{eq-dominate}
\forall k \geq 0, \quad \Prob[ T \geq k n] \leq \gamma_{\I}^{kn}.	
\end{align}
This implies $T$ is stochastically dominated by a geometric distribution. 
Each iteration of the \while{} loop terminates within finite number of steps. 
Thus, the algorithm terminates with probability~1.	

\subsection{Correctness upon termination}
We show that upon termination, the output follows the correct distribution.
Let  $(\X, \R) \in [q]^V \times 2^V$ be the random pair maintained by the algorithm.
The following condition is the ``loop invariant'' of the random pair $(\X, \R)$.

\begin{condition}[conditional Gibbs property]
\label{condition-invariant}
For any $R \subseteq V$ and $\sigma \in [q]^R$,
conditioned on $\R = R$ and $X_R = \sigma$, 
the random configuration $X_{V \setminus R}$ follows the law $\mu_{\I}^\sigma$.
\end{condition}


\Cref{condition-invariant} is satisfied initially by the initial pair $(\X, \R) = (\X, V)$.
Furthermore, consider the \while{} loop that transforms
\begin{align*}
(\X,\R) \rightarrow (\X', \R').
\end{align*}
Then next lemma shows that \Cref{condition-invariant} holds inductively assuming \Cref{property-cond-ind}.
\begin{lemma}
\label{lemma-detailed-invariant}
Suppose that  $(\X, \R) \in [q]^V \times 2^V$ is a random pair such that $\X$ is feasible and the pair $(\X, \R)$ satisfies Condition~\ref{condition-invariant}.
Then, the random pair $(\X',\R')$ satisfies Condition~\ref{condition-invariant}. 
\end{lemma}
\noindent
By \Cref{observation-positive}, the random configuration $\X \in [q]^V$ maintained by algorithm is always feasible. 
\Cref{lemma-detailed-invariant} guarantees that \Cref{condition-invariant} holds all along for the random pair $(\X, \R)$ maintained by \Cref{alg:perfect-sampler-gen}.
In particular, when the algorithm terminates, $\mathcal{R} = \emptyset$ and  the output $\X$ follows the correct distribution $\mu_{\I}$. 
This proves Theorem~\ref{theorem-correctness}.


\begin{proof}[Proof of \Cref{lemma-detailed-invariant}]
It is sufficient to prove that for any $R \subseteq V$, any feasible partial configuration $\rho \in [q]^R$ and any vertex $u \in R$, 
conditioned on $\R = R $, $X_R = \rho$, and the vertex picked in \Cref{line-pick} being $u$, the new random pair $(\X',\R')$ after one iteration of the \while{} loop satisfies Condition~\ref{condition-invariant}. 

Fix  $R \subseteq V$ and a feasible partial configuration $\rho \in [q]^R$.
Let $\*u \in R$ denote the uniform random vertex picked in \Cref{line-pick}.
Fix a vertex $u \in R$.  
Let $\mathcal{E}$ denote the event
\begin{align*}
  \mathcal{E}: \quad X_R = \rho \land \R = R \land \*u = u.
\end{align*}
Since $(\X, \R)$ satisfies \Cref{condition-invariant} and given the set $R$, $\*u$ is independent from $\X$, we have
\begin{align}
\label{eq-X-V-setminus-R}
\forall \tau \in [q]^{V \setminus R}:\quad
\Prob\left[ X_{V \setminus R} = \tau \mid \mathcal{E} \right]	= \mu_{\I}^\rho(\tau).
\end{align}
Recall that $\+F$ is the Bayes filter.
Depending on whether $\+F$ succeeds or not, we have two cases.

The easier case is  when $\+F$ fails.
Recall that the set $B$ is fixed by $R$ and $u$. In this case, $\R'= R \cup \partial B$ and $\X' = \X$.
Conditioned on $\+E$, we know that $X_u = \rho_u$ and  $X_{\partial B \cap R} = \rho_{\partial B \cap R}$,
the filter $\+F$ depends only on the partial configuration $X_{\partial B \setminus R}$. 
For any configuration $\sigma \in [q]^{\partial B \setminus R}$, 
conditioned on $\mathcal{E}$ and $X_{\partial B \setminus R} = \sigma$,
the failure of $\mathcal{F}$ is independent from $X_{V \setminus (R \cup \partial B)}=X_{V\setminus \R'}$.
Thus, by~\eqref{eq-X-V-setminus-R}, conditioned on $\mathcal{E}$, $X_{\partial B \setminus R} = \sigma$ and the failure of $\+F$, 
we have that $X'_{V \setminus \R'} = X_{V\setminus \R'} \sim \mu_{\I}^{\rho \uplus \sigma}$, i.e.~$(\X', \R')$ satisfies \Cref{condition-invariant}.

Now we analyze the main case that $\+F$ succeeds. 
If this case does occur, we must have
\begin{align}
\label{eq-proof-assume-mu-min}
\mu_{\min}(R,u,\rho)\triangleq \min\{ \mu^{\sigma}_{u,\I}(\rho_u) \mid \sigma \in [q]^{\partial B} \text{ s.t. } \sigma_{R \cap \partial B } = \rho_{R \cap \partial B}\} > 0.
\end{align}

Define $R_u \triangleq R \setminus \{u\}$.
The fact that $\+F$ succeeds means $\R' = R \setminus \{u\} = R_u$ and $X'_{R_u}=X_{R_u} = \rho_{R_u}$. 
Hence, we only need to show that
\begin{align}
  \label{eq-X-V-setminus-Ru}
  \forall \tau \in [q]^{V \setminus R_u}:\quad
  \Prob\left[ X'_{V \setminus R_u} = \tau \mid \mathcal{E} \land \Fss \right]	=  \mu_{\I}^{\rho(R_u)}(\tau).
\end{align}
Recall $ B = (B_{\ell}(u) \setminus R) \cup \{u\} = B_{\ell}(u) \setminus R_u$.
We define the following set:
\begin{align*}
  H \triangleq V \setminus \left\{ R_u \cup B \right\}.
\end{align*}
Namely, $B$ is the set whose configuration is resampled, and $H$ is the set whose configuration is untouched, i.e.~$X'_H=X_H$.
Note that $B \uplus H \uplus R_u = V$. By the chain rule:
\begin{align}\label{eq:chain-rule-success}
  &\Prob\left[X'_{V \setminus R_u}=\tau \land \Fss \mid \mathcal{E}\right] 
  = \Prob\left[ X'_H = \tau_H\land X'_B = \tau_B\land \Fss \mid \+E \right]\\
   =\,& \Prob[X'_H = \tau_H \mid \mathcal{E}]\cdot \Prob[\Fss \mid \mathcal{E}\land X'_H = \tau_H]\cdot 
  \Prob[X'_B = \tau_B \mid \mathcal{E}\land X'_H = \tau_H\land \Fss].\notag
\end{align}
As $X'_H=X_H$, \eqref{eq-X-V-setminus-R} implies that 
\begin{align*}
  \Prob[X'_H = \tau_H \mid \mathcal{E}] = \mu_{H,\I}^\rho(\tau_H).
\end{align*}
By Line~\ref{line-sample} of Algorithm~\ref{alg:perfect-sampler-gen},
$X'_B$ is redrawn from the distribution
$\mu_{B,\I}^{X_{\partial B}}(\cdot)$.
By conditional independence property (\Cref{property-cond-ind}), we have $\mu_{B,\I}^{X_{\partial B}}(\cdot) = \mu_{\I}^{X_{V \setminus B}}(\cdot)$.
Note that $V \setminus B = R_u \uplus H$.
Conditioned on $\+E\land X'_H = \tau_H$, $X_{R_u} = \rho_{R_u}$ and $X_{H}=X'_H=\tau_H$, thus $\mu_{B, \I}^{X_{\partial B}}(\cdot) = \mu_{\I}^{\rho(R_u)\uplus \tau(H)}(\cdot)$.
Hence,
\begin{align}
  \eqref{eq:chain-rule-success}
  =\mu_{H,\I}^\rho(\tau_H)\cdot \mu_{\I}^{\rho(R_u)\uplus \tau(H)}(\tau_B)
  \cdot \Prob[\Fss \mid \mathcal{E}\land X'_H = \tau_H].  \label{eq-chain-rule}
\end{align}

To finish the proof, we need to calculate $\Prob[\Fss\mid \mathcal{E}\land X'_H = \tau_H]$.
This is done by the following claim, whose proof is deferred to the end of the section.
\begin{claim}
  \label{claim-correctness}
  Assume~\eqref{eq-proof-assume-mu-min}. It holds that
  \begin{align}
    \label{eq-claim-1}
     \mu_{H, \I}^{\rho}(\tau_H) > 0 \quad \Longleftrightarrow \quad \mu_{H, \I}^{\rho(R_u)}(\tau_H) > 0.
  \end{align}
  Furthermore, for $\tau_H$ such that $\Prob[X'_H = \tau_H \mid \mathcal{E}] = \mu_{H, \I}^{\rho}(\tau_H) > 0$,
  \begin{align}
    \label{eq-claim-2}
    \Prob[\Fss\mid \mathcal{E}\land X'_H = \tau_H] = C\cdot \frac{\mu_{H, \I}^{\rho(R_u)}(\tau_H)}{\mu_{H,\I}^\rho(\tau_H)},
  \end{align}
where $C = C(R, u, \rho) > 0$ is a constant depending only on $R, u, \rho$ but not on $\tau$.
\end{claim}

Combining~\eqref{eq-chain-rule} and \Cref{claim-correctness}, we have
\begin{align}
\label{eq-final-prob}
\forall \tau \in [q]^{V \setminus R_u}, \quad 	\Prob\left[X'_{V \setminus R_u}=\tau \land \Fss\mid \mathcal{E}\right]	= C\cdot\mu^{\rho(R_u)}_{\I}(\tau).
\end{align}
This equation can be verified in two cases:
\begin{itemize}
\item If $\mu_{H,\I}^\rho(\tau_H)  = 0$, then by~\eqref{eq-claim-1}, $\mu_{H, \I}^{\rho(R_u)}(\tau_H) = 0$, thus $\text{LHS} = \text{RHS} = 0$.
\item If $\mu_{H,\I}^\rho(\tau_H)  > 0$, by~\eqref{eq-chain-rule}~and~\eqref{eq-claim-2}, we have $\text{LHS} = C \cdot \mu_{H, \I}^{\rho(R_u)}(\tau_H)\cdot  \mu_{\I}^{\rho(R_u)\uplus \tau(H)}(\tau_B) = \text{RHS}$, where the last equation holds because $\tau \in [q]^{V \setminus R_u}$ and $V \setminus R_u= {H \uplus B}$.
\end{itemize}
Thus, the probability that $\Fss$ is 
\begin{align}
\label{eq-succeeds-prob}	
\Prob\left[ \Fss\mid \mathcal{E}\right] = \sum_{\sigma \in [q]^{V \setminus R_u }} \Prob\left[X'_{V \setminus R_u}=\sigma \land \Fss\mid \mathcal{E}\right] = \sum_{\sigma \in [q]^{V \setminus R_u }}C\cdot\mu^{\rho(R_u)}_{\I}(\sigma) =C,
\end{align}
where the last equation holds because $\sum_{\sigma \in [q]^{V \setminus R_u }}\mu^{\rho(R_u)}_{\I}(\sigma) = 1$ and $C = C(R, u, \rho) > 0$ is a constant depending only on $R, u, \rho$.
Thus, for any $\tau \in [q]^{V \setminus R_u}$, combining~\eqref{eq-final-prob} and~\eqref{eq-succeeds-prob}, we have
\begin{align*}
 \Prob\left[X'_{V \setminus R_u}=\tau \mid \Fss \land \mathcal{E}\right]=\frac{\Prob\left[X'_{V \setminus R_u}=\tau \land \Fss\mid \mathcal{E}\right]}{\Prob\left[ \Fss\mid \mathcal{E}\right]}
=\frac{C \cdot \mu^{\rho(R_u)}_{\I}(\tau)}{C} = \mu^{\rho(R_u)}_{\I}(\tau),
\end{align*}
where the last equation holds due to $C = C(R,u,\rho) > 0$. This proves~\eqref{eq-X-V-setminus-Ru}.
\end{proof}

\begin{proof}[Proof of Claim~\ref{claim-correctness}]
We first introduce the following definitions. 
Recall that $R_u \uplus B \uplus H = V$. We further partition $\partial B$ into two disjoint sets $\partial B \cap R_u$ and $\partial B \setminus R_u$. Define
\begin{equation}
\label{eq-def-part-B}
\begin{split}
S &\triangleq \partial B \setminus R_u = \partial B \cap H,\\
\Psi &\triangleq \partial B \cap R_u = \partial B \cap R.
\end{split}
\end{equation}

We now prove~\eqref{eq-claim-1}.
Since $\rho= \rho(R_u) \uplus \rho(u)$, by the Bayes law, we have the following relation between $\mu_{H, \I}^{\rho(R_u)}(\tau_{H})$ and $\mu_{H, \I}^{\rho}(\tau_{H})$:
\begin{align}
\label{eq-proof-bayes}
\mu_{H,\I}^\rho(\tau_{H}) = 	\mu_{H,\I}^{\rho(R_u)\uplus \rho(u)}(\tau_{H}) = \frac{\mu_{u,\I}^{\rho(R_u) \uplus \tau(H)}(\rho_u) }{ \mu^{\rho(R_u)}_{u,\I}(\rho_u)} \cdot \mu^{\rho(R_u)}_{H,\I}(\tau_{H}).
\end{align}
Note that $\rho \in [q]^R$ is a feasible configuration, thus $\mu^{\rho(R_u)}_{u,\I}(\rho_u) > 0$ and the above ratio is well-defined. Note that $u \in B$ and $R_u \uplus B \uplus H = V$. The set $\partial B$ separates $u$ from $(R_u \uplus H) \setminus \partial B$. Note that $\partial B = S \uplus \Psi$, where $S$ and $\Psi$ is defined in~\eqref{eq-def-part-B}. By the conditional independence property (\Cref{property-cond-ind}), we have
\begin{align}
\label{eq-proof-cond-ind}
\mu_{u,\I}^{\rho(R_u) \uplus \tau(H)}(\rho_u) = 	\mu_{u,\I}^{\rho(\Psi) \uplus \tau(S)}(\rho_u) \geq \mu_{\min}(R,u,\rho) > 0,
\end{align}
where the first inequality is because $\mu_{\min}(R,u,\rho)$ in~\eqref{eq-proof-assume-mu-min} can be rewritten as $\min_{\eta \in [q]^{ S }}\mu_{u,\I}^{\rho(\Psi) \uplus \eta}(\rho_u)$, and the second inequality is because  $\mu_{\min}(R,u,\rho) > 0$ due to the lower bound in~\eqref{eq-proof-assume-mu-min}.


Next, we prove~\eqref{eq-claim-2}. Suppose $\mu_{H, \I}^{\rho}(\tau_H) > 0$. 
Combining~\eqref{eq-proof-bayes} and~\eqref{eq-proof-cond-ind}, it remains to prove that 
\begin{align}
\label{eq-proof-new-target}
  \Prob[\Fss\mid \mathcal{E}\land X'_H = \tau_H] = C\cdot \frac{\mu_{H, \I}^{\rho(R_u)}(\tau_H)}{\mu_{H,\I}^\rho(\tau_H)} = C\cdot \frac{\mu^{\rho(R_u)}_{u,\I}(\rho_u) }{\mu_{u,\I}^{\rho(\Psi) \uplus \tau(S)}(\rho_u)}.
\end{align}
Conditional on $\+E$, we have $X_{\Psi} = \rho_{\Psi}$ and $X_u = \rho_u$.
Recall that $X'_H=X_H$, $S \subseteq H$ and $S \uplus \Psi =\partial B$.
Conditional on $X'_H = \tau_H$, it holds that $X_S = \tau_S$.
By the definition of the filter $\mathcal{F}$ in Line~\ref{line-F} of \Cref{alg:perfect-sampler-gen}, we have that
\begin{align}
  \label{eq-original-P}
  \Prob\left[ \Fss \mid \mathcal{E}  \land X'_H = \tau_H\right]  
  & = \frac{\mu_{\min}(\+R,u,X_{\+R})}{\mu_{u,\I}^{X_{\partial B}}(X_u)} = \frac{\mu_{\min}(R,u,\rho)}{\mu_{u,\I}^{\rho(\Psi) \uplus \tau(S)}(\rho_u)}.
\end{align}
Combining~\eqref{eq-original-P} and~\eqref{eq-proof-new-target}, we can set $ C = C(R,u,\rho)$ in~\eqref{eq-proof-new-target} as
\begin{align}
\label{eq-set-C}
	C = C(R,u,\rho) \triangleq \frac{\mu_{\min}( R,u,\rho) }{ \mu_{u,\I}^{\rho(R_u)} (\rho_u) } = \frac{1}{\mu_{u,\I}^{\rho(R_u)} (\rho_u)}\cdot \min_{\eta \in [q]^{ S }}\mu_{u,\I}^{\rho(\Psi) \uplus \eta}(\rho_u).
\end{align}
Note that $\mu^{\rho(R_u)}_{u,\I}(\rho_u) > 0$ because $\rho$ is feasible, and $\mu_{\min}( R,u,\rho) > 0$ due to the lower bound in~\eqref{eq-proof-assume-mu-min}.  This implies $C(R,u,\rho) > 0$.
Remark the sets $S$ and $\Psi$ are determined by $R$ and $u$. 
This implies that the $C(R,u,\rho)$ defined as above depends only on $R,u,\rho$.
This proves \eqref{eq-claim-2}.
\end{proof}

\section{Efficiency under strong spatial mixing}
\label{section-running-time}
In this section, 
we prove the following result.
\begin{condition}
\label{condition-SSM-ratio}
Let $\I=(G, [q], \boldsymbol{b}, \boldsymbol{A})$ be a permissive spin system where $G=(V,E)$.
There is an integer $\ell=\ell(q)\ge 2$ such that the following holds:
for every $v \in V$, $\Lambda \subseteq V$, and any two partial configurations $\sigma, \tau \in [q]^\Lambda$ satisfying $\min\{\dist_G(v, u) \mid u \in \Lambda,\ \sigma_u\neq\tau_u\}=\ell$,
\begin{align}
\label{eq:condition-SSM}
  \forall a \in [q]: \quad \left\vert \frac{\mu_{v, \I}^\sigma(a)}{ \mu_{v,\I}^\tau(a) } - 1 \right\vert \leq \frac{ 1 }{5 \left\vert S_{\ell}(v) \right\vert}
  \quad(\text{with the convention $0/0 = 1$}),
\end{align}
where $S_{\ell}(v)\triangleq \{u \in V \mid \dist_G(v, u) = \ell\}$ denotes the {sphere} of radius $\ell$ centered at $v$ in $G$.
\end{condition}

\begin{theorem}
\label{theorem-general-ratio}
Let $\mathfrak{I}$ be a class of permissive spin systems satisfying Condition~\ref{condition-SSM-ratio}.
Given any instance $\I=(G,[q],\boldsymbol{b},\boldsymbol{A})\in\mathfrak{I}$, the \Cref{alg:perfect-sampler-gen} with parameter $\ell = \ell^* - 1$ outputs a perfect sample from $\mu_{\I}$ within $ O\left(n\cdot q^{2\Delta^{\ell^*}}\right) $ time in expectation, 
where $n$ is the number of vertices in $G$, $\Delta$ is the maximum degree of $G$, $\ell^* = \ell^*(q) \geq 2$ is determined by \Cref{condition-SSM-ratio}, and $O(\cdot)$ hides only absolute constants.
\end{theorem}
The correctness part of \Cref{theorem-general-ratio} follows from \Cref{theorem-correctness}, we focus on the efficiency part.
The proof sketch is that if $\mathfrak{I}$ satisfies \Cref{condition-SSM-ratio} with parameter $\ell^*$, we set the parameter $\ell$ in \Cref{alg:perfect-sampler-gen} as $\ell = \ell^* - 1$. We prove that after each iteration of the \while{} loop, the size of $\+R$ decreases by at least $\frac{1}{5}$ in expectation. 
Note that the initial $\+R = V$, thus the initial size of $\+R$ is $n$. 
By the optional stopping theorem, the number of iterations of the \while{} loop is $O(n)$ in expectation.
One can verify that the time complexity of  the \while{} loop is $O(q^{2\Delta^{\ell+1}}) = O(q^{2\Delta^{\ell^*}})$. 
This proves the running time result.

To analyze the expected size of $\+R$ after each  iteration of the \textbf{while} loop, we analyze the Bayes filter $\+F$ in \Cref{line-F}. The probability that $\+F$ fails is $1 - \mu_{\min}/ \mu_{u}^{X_{\partial B}}(X_u)$. Suppose $\sigma \in [q]^{\partial B}$ achieves $\mu_{\min} = \mu_{u}^\sigma(X_u)$. By \Cref{fact-0}, we can verify that $\sigma$ and $X_{\partial B}$ can be differ only at $\partial B \setminus \+R \subseteq S_{\ell+1}(u) = S_{\ell^*}(u)$. By \Cref{condition-SSM-ratio}, we know that $\Prob[\+F \text{ fails}] \leq \frac{1}{5\abs{S_{\ell^*}(u)}}$. In addition, we have
\begin{itemize}
\item if $\Fss$, the size of $\+R$ decreases by 1;
\item if $\+F$ fails, the size of $\+R$ increases by $\abs{\partial B \setminus  \+R}$, it easy to verify $\partial B \setminus  \+R  \subseteq S_{\ell^*}(u)$ by \Cref{fact-0}, thus, the size of $\+R$ increases by at most $\abs{S_{\ell^*}(u)}$.
\end{itemize}
Thus, after each iteration of the \while{} loop, the size of $\+R$ decreases by at least $\frac{1}{5}$ in expectation.

In the formal proof, we actually prove a stronger result. We first introduce the following condition. 
\begin{condition}
\label{condition-SSM-weak}
Let $\I=(G, [q], \boldsymbol{b}, \boldsymbol{A})$ be a permissive spin system where $G=(V,E)$.
There is an integer $\ell=\ell(q)\ge 2$ such that the following holds:
for every $v \in V$, any two disjoint sets $A, B \subseteq V$ with $\dist_G(v,B) = \min\{\dist_G(v,u) \mid u \in B\} = \ell$, and any partial configuration $\sigma \in [q]^A$,
\begin{align}
\label{eq:condition-SSM-weak}
\forall a \in [q],\quad 1 - \frac{\min_{\tau \in [q]^B} \mu_{v, \I}^{\sigma \uplus \tau}(a)}{ \mu_{v,\I}^\sigma(a) } \leq \frac{ 1 }{5 \left\vert S_{\ell}(v) \right\vert}
  \quad(\text{with the convention $0/0 = 1$}),
\end{align}
where $S_{\ell}(v)\triangleq \{u \in V \mid \dist_G(v, u) = \ell\}$ denotes the {sphere} of radius $\ell$ centered at $v$ in $G$.
\end{condition}
It is straightforward to verify that \Cref{condition-SSM-ratio} implies \Cref{condition-SSM-weak}. In the rest of this section, we prove that the  efficiency result in \Cref{theorem-general-ratio} holds under \Cref{condition-SSM-weak}.

Let $\mathcal{I} = (G, [q], \boldsymbol{b}, \boldsymbol{A}) \in \mathfrak{I}$ be the input instance
satisfying \Cref{condition-SSM-weak} with some $\ell^* \geq 2$.
Set the parameter $\ell$ in Algorithm~\ref{alg:perfect-sampler-gen} as $\ell = \ell^* - 1$.
%
Denote by $T$ the number of iterations of the \while{} loop in \Cref{alg:perfect-sampler-gen}. 
To prove the  efficiency result in Theorem~\ref{theorem-general-ratio}, 
we bound the maximum running time of the \while{} loop and the expectation of $T$ in the following two lemmas.

\begin{lemma}
\label{lemma-running-time-alg2}
The running time of each \while{} loop is at most \twoBtime $= O(q^{2\Delta^{\ell^*}})$.
\end{lemma}

\begin{lemma}
\label{lemma-ET}
$\E{T} \leq 5n$.	
\end{lemma}

Since the input instance $\I$ is permissive (\Cref{definition-locally-admissible})
, the initial feasible configuration can be constructed by a simple greedy algorithm.
The running time of the first two lines in Algorithm~\ref{alg:perfect-sampler-gen} is $O(n\Delta )$. 
Combining this with Lemma~\ref{lemma-running-time-alg2} and Lemma~\ref{lemma-ET} proves the  efficiency result in Theorem~\ref{theorem-general-ratio}.
%

\begin{proof}[Proof of Lemma~\ref{lemma-running-time-alg2}]
  We first show that $\mu_{u,\I}^{X_{\partial B}}(X_u)$ can be computed in time \Btime,
  where $O(\cdot)$ hides an absolute constant.
  Let $\widetilde{G}=G[B \cup \partial B]$ and $\widetilde{\I}$ be the instance restricted to $\widetilde{G}$.
  By the conditional independence property (\Cref{property-cond-ind}), it is straightforward to verify
  \begin{align}
    \label{eq-K-reduced}
    \mu_{u,\I}^{X_{\partial B}}(X_u) = \mu_{u,\widetilde{\I}}^{X_{\partial B}}(X_u) 
  \end{align}
  since $\partial B$ separates $B$ from $V\setminus (B\cup \partial B)$ and $u \in B$.
  Since $\abs{B}\leq\abs{B_{\ell}(u)}\le \frac{\Delta^{\ell+1}-1}{\Delta-1}\le \Delta^{\ell+1}$,
  it takes at most \Btime\ to enumerate all possibilities and compute $ \mu_{u,\I}^{X_{\partial B}}(X_u)$ using~\eqref{eq-K-reduced}.
By \Cref{fact-0}, $\partial B \subseteq S_{\ell+1}(u) \cup R$.
  Since $\abs{\partial B \setminus R } \le \abs{S_{\ell+1}(u)}\le \Delta^{\ell+1}$,
  we can enumerate all $[q]^{\partial B \setminus R}$ to compute $\mu_{\min}$ in time \twoBtime.
  The total running time for the first three lines of the \while{} loop is at most \twoBtime.

  Another non-trivial computation is to sample $X(B)$ from $\mu_{\I}^{X_{\partial B}}$.
  Similar to \eqref{eq-K-reduced},
  conditional independence implies that this can be done by straightforward enumeration in time \Btime.
  The total running time of the \while{} loop is thus \twoBtime$= O(q^{2\Delta^{\ell^*}})$.
\end{proof}

\begin{proof}[Proof of Lemma~\ref{lemma-ET}]
Define a sequence of random pairs $(\X_0,\R_0),(\X_1,\R_1),\ldots,(\X_T,\R_T)$, where each $(\X_t,\R_t) \in [q]^V \times 2^V$. 
The initial $(\X_0,\R_0)$ is constructed by the first two lines of Algorithm~\ref{alg:perfect-sampler-gen}.
In $t$-th \while{} loop, Algorithm~\ref{alg:perfect-sampler-gen} updates $(\X_{t-1},\R_{t-1}) $ to $(\X_t,\R_t)$. 
For any $t \geq 0$, we use a random variable $Y_t \triangleq \abs{\R_t}$ to denote the size of $\R_t$.
The stopping time $T$ is the smallest integer such that $Y_t=0$. 

Define the \emph{execution log} of \Cref{alg:perfect-sampler-gen} up to time $t$ as
\begin{align*}
\+L_t \triangleq  (X_0(\R_0), \R_0), (X_1(\R_1), \R_1),\ldots, (X_t(\R_t), \R_t).
\end{align*}
Note that the algorithm terminates at time $T$ if and only if $\R_T = \emptyset$. 
In the $t$-th iteration of the \while{} loop, we use $\+F_t$ to denote the Bayes filter and $\*u_t$ to denote the random vertex picked in~\Cref{line-pick}.
We have the following claim.
\begin{claim}
\label{claim-lower-bound}

Given any execution log $\+L_{t-1}$ created by \Cref{alg:perfect-sampler-gen} such that $\R_{t-1}\neq \emptyset$, for any $u \in \R_{t-1}$,
\begin{align*}
\Prob[\+F_t \text{ succeeds} \mid \+L_{t-1}, \*u_t =u] \geq \begin{cases}
1 &\text{if } |S_t| = \emptyset;\\
 1- \frac{2}{5|S_t|}&\text{if }|S_t| \neq \emptyset,
 \end{cases}
\end{align*}
where $S_t = \partial B_t \setminus \+R_{t-1}$ and $B_t = (B_{\ell}(u) \setminus \R_{t-1})\cup\{u\}$ is the set $B$ in the $t$-th iteration the \while{} loop.
\end{claim}

Note that given $\+L_{t-1}$, the vertex $\*u_t \in \R_{t-1}$ is sampled uniformly and independently.
For any $1\leq  t \leq T$ and any execution log $\+L_{t-1}$ created by \Cref{alg:perfect-sampler-gen}, if $S_t = \emptyset$, by \Cref{claim-lower-bound}, we have
\begin{align*}
\E{Y_t \mid \+L_{t-1}, S_t = \emptyset}	= Y_{t-1} - 1;
\end{align*}
Suppose $S_t \neq \emptyset$. If $\+F_t$ fails, then $\R_{t} = \R_{t-1} \cup \partial B_t = \R_{t-1} \cup (\partial B_t \setminus \R_{t-1})$.
In other words, $|S_t|$ new vertices will be added into $\R_{t-1}$ if $\+F_t$ fails.
We have
\begin{align*}
  \E{Y_t \mid \+L_{t-1}, S_t \neq \emptyset}	 &\leq Y_{t-1} -\Prob[\mathcal{F}_t \text{ succeeds} \mid \+L_{t-1},S_t \neq \emptyset] + \Prob[\mathcal{F}_t \text{ fails} \mid \+L_{t-1},S_t \neq \emptyset] \cdot |S_t|\\
  &\leq Y_{t-1} + \frac{2}{5|S_t|} - \frac{3}{5} \tag{by \Cref{claim-lower-bound}}\\ 
  &\leq Y_{t-1} - \frac{1}{5}.
\end{align*}
Combining two cases together implies
\begin{align*}
 \E{Y_t \mid \+L_{t-1}} = \E{Y_t\mid (X_0(\R_0), \R_0), (X_1(\R_1), \R_1),\ldots, (X_{t-1}(\R_{t-1}), \R_{t-1}) } \leq Y_{t-1} - \frac{1}{5}.	
\end{align*}
We now define a sequence $Y'_0,Y'_1,\ldots,Y'_T$ where each $Y'_t = Y_t + \frac{t}{5}$. Thus, $Y'_0,Y'_1,\ldots,Y'_T$ is a super-martingale with respect to $(X_0(\R_0), \R_0), (X_1(\R_1), \R_1),\ldots, (X_T(\R_T), \R_T)$ and $T$ is a stopping time. Note that each $|Y'_t - Y'_{t-1}| \leq n + 1$ is bounded and $\E{T}$ is finite due to~\eqref{eq-dominate}. Due to the optional stopping theorem~\cite[Chapter 5]{durrett2019probability} , we have $\E{Y'_T} \leq \E{Y'_0} = \E{Y_0}$. Hence 
\begin{align*}
\E{T} \leq 5\E{Y_0} = 5n,
\end{align*}
where the last eqaution is because $\E{Y_0} = \E{|\R_0|} = n$.	
\end{proof}

\begin{proof}[Proof of \Cref{claim-lower-bound}]
Suppose $S_t = \partial B_t \setminus \+R_{t-1} = \emptyset$.
This implies $\partial B_t \subseteq \R_{t-1}$.
By the definition of $\mu_{\min}$, we have $\mu^{X_{t-1}(\partial B_t)}_{u,\I}(X_{t-1}(u)) = \mu_{\min}$ and $\Prob[\+F_t \text{ succeeds}] = 1$.
In the following proof, we  assume $S_t \neq \emptyset$.

We need the following property to prove the claim.
Fix an execution log $L_{t}$ up to time $t \geq 0$:
\begin{align*}
L_{t} = (\rho_0, R_0), (\rho_1,R_1),\ldots, (\rho_{t}, R_{t}),	
\end{align*}
where each $R_i \subseteq V$, each $\rho_i \in [q]^{R_i}$. Assume $R_{t} \neq \emptyset$ and $L_{t}$ is a feasible execution log, i.e. $\Prob[\+L_{t} = L_{t}]	 > 0$.
We claim that given the log $L_t$, the random $\X_t \in [q]^V$ satisfies $X_t(R_t) = \rho_t$ and
\begin{align}
\label{eq-martingale}
\forall \tau \in [q]^{V \setminus R_{t}}, \quad \Prob[X_{t}(V \setminus R_{t}) = \tau \mid \+L_{t}=L_{t}] = \mu_{\I}^{\rho_{t}}(\tau).	
\end{align}
\Cref{eq-martingale} is proved by the induction on $t$. If $t = 0$, we have $R_0 = V$, \Cref{eq-martingale} holds trivially. Assume \Cref{eq-martingale} holds for all $t < k$.  Fix any feasible execution log $L_{k} = (\rho_0, R_0), (\rho_1,R_1),\ldots, (\rho_{k}, R_{k})$ such that $R_k \neq \emptyset$. 
	Since $L_k$ is feasible, we have $R_{k-1} \neq \emptyset$. Consider the $k$-th iteration of the \while{} loop. The $k$-th iteration  exists because $R_{k-1} \neq \emptyset$. By induction hypothesis, conditioning on the execution log $\+L_{k-1} = (\rho_0, R_0), (\rho_1,R_1),\ldots, (\rho_{k-1}, R_{k-1})$, the random pair $(\X_{k-1},\R_{k-1})$ satisfies the \Cref{condition-invariant} and $\X_{k-1}$ is a feasible configuration (since $L_k$ is a feasible execution log, thus $\rho_{k-1}$ is feasible).  By \Cref{lemma-detailed-invariant}, conditioning on the execution log $\+L_{k-1} = (\rho_0, R_0), (\rho_1,R_1),\ldots, (\rho_{k-1}, R_{k-1})$, the random pair $(\X_{k},\R_{k})$ satisfies the \Cref{condition-invariant}. By \Cref{condition-invariant}, assuming the further condition $\R_{k} = R_k$ and $X_k(R_k) = \rho_k$, \Cref{eq-martingale} holds for $t = k$. This proves  \Cref{eq-martingale}.

Consider a feasible execution log up to time $t - 1 \geq 0$:
\begin{align*}
L_{t-1} = (\rho_0, R_0), (\rho_1,R_1),\ldots, (\rho_{t-1}, R_{t-1})
\end{align*}
satisfying $R_{t-1} \neq \emptyset$, where each $R_i \subseteq V$ and each $\rho_i \in [q]^{R_i}$.
Given the execution log $\+L_{t-1} = L_{t-1}$, we fix a vertex $u \in R_{t-1}$ and assume $\*u_t = u$.
We analyze the $t$-th iteration of the \while{} loop.
To simplify the notation, we drop the index and denote
\begin{align*}
\X = \X_{t-1},\quad R= R_{t-1},\quad\rho = \rho_{t-1},\quad B = B_t = (B_{\ell}(u) \setminus R )\cup \{u\} ,\quad S = S_t = \partial B\setminus R.
\end{align*}
Note that the vertex $\*u_t$ is sampled from $R$ uniformly and independently.
By~\eqref{eq-martingale}, given $\+L_{t-1}=L_{t-1}$ and $\*u_t = u$, it holds that $X(R) = \rho$ and $(\X, R)$ satisfies \Cref{condition-invariant}.
By \Cref{property-lower-bound}, we know $\mu_{\min}(R,u,\rho) > 0$, thus the lower bound in~\eqref{eq-proof-assume-mu-min} holds.
According to the proof of \Cref{lemma-detailed-invariant}, combining~\eqref{eq-succeeds-prob} and~\eqref{eq-set-C}, we have
\begin{align*}
\Pr[\+F_t \text{ succeeds} \mid \+L_{t-1} = L_{t-1}, \*u_t = u] = \frac{1}{\mu_{u,\I}^{\rho(R_u)} (\rho_u)}\cdot \min_{\eta \in [q]^{ S }}\mu_{u,\I}^{\rho(\Psi) \uplus \eta}(\rho_u),	
\end{align*}
where $R_u = R \setminus \{u\}$, $S = \partial B \setminus R$ and $\Psi = \partial B \cap R$.
Note that $\partial B = S \uplus \Psi$ and $u \in B$,  the set $\partial B$ separates  $u$ from $V \setminus (B \cup \partial B)$.
Since $\Psi \subseteq R_u$ and two sets $R_u$ and $B$ are disjoint,
by the conditional independence property (\Cref{property-cond-ind}), we have
$\mu_{u,\I}^{\rho(\Psi) \uplus \eta}(\rho_u) = \mu_{u,\I}^{\rho(R_u) \uplus \eta}(\rho_u)$.
This implies
\begin{align*}
\Pr[\+F_t \text{ succeeds} \mid \+L_{t-1} = L_{t-1}, \*u_t = u] = \frac{1}{\mu_{u,\I}^{\rho(R_u)} (\rho_u)}\cdot \min_{\eta \in [q]^{ S }}\mu_{u,\I}^{\rho(R_u) \uplus \eta}(\rho_u).
\end{align*}
By \Cref{fact-0}, we have $S = \partial B \setminus R \subseteq S_{\ell + 1}(u)$.
We take $A = R_u, B = S, v= u, \sigma = \rho(R_u)$ and $a = \rho_u$ in \Cref{condition-SSM-weak}, since $\dist_G(u, S) = \ell + 1 = \ell^*$ and $|S| \leq |S_{\ell^*}(u)|$, this proves that
\begin{align*}
\Pr[\+F_t \text{ succeeds} \mid \+L_{t-1} = L_{t-1}, \*u_t = u]	\geq 1 - \frac{1}{5|S_{\ell^*}(u)|} \geq 1 - \frac{1}{5|S|} \geq 	 1 - \frac{2}{5|S|}. 
\end{align*}
\end{proof}

\begin{remark}
\label{remark-improve}
Suppose the input instance from the class $\mathfrak{I}$ satisfies~\cref{condition-SSM-ratio} with some $\ell^* \geq 2$ and take $\ell = \ell^* - 1$ in \Cref{alg:perfect-sampler-gen}.
We could tweak \Cref{alg:perfect-sampler-gen} to reduce its running time to $O\left( n\cdot q^{ \Delta^{\ell^*} }  \right)$.
Let $\Psi \triangleq \partial B \cap \R$ and $S \triangleq \partial B \setminus \R$.
Note that $S \subseteq S_{\ell^*}(u)$ by \Cref{fact-0}.
  The idea is that instead of calculating $\mu_{\min}$, we may simply compute $\mu_{u,\I}^{X(\Psi) \uplus \sigma}(X_u)$ where $\sigma = \boldsymbol{1} \in [q]^{S}$ is a one-vector,   then use~\eqref{eq:condition-SSM} to get a lower bound $\mu_{\mathrm{low}}$ of $\mu_{\min}$ as
  \begin{align*}
  \mu_{\mathrm{low}} \triangleq  	\left(1 - \frac{ 1 }{5 \left\vert S_{\ell^*}(u) \right\vert}\right)\mu_{u,\I}^{X(\Psi) \uplus \sigma}(X_u).
  \end{align*}
\noindent 
By \Cref{condition-SSM-ratio}, $\mu_{\mathrm{low}} \leq \mu_{u,\I}^{X(\partial B)}(X_u)$.
Then we use $\mu_{\mathrm{low}}$ instead of $\mu_{\min}$ in the definition of $\+F$.
  It is straightforward to check that \Cref{alg:perfect-sampler-gen} is still correct with this tweak, and for each iteration of the \while{} loop, given any $\X, \R$, it holds that
  \begin{align*}
\Prob[\Fss \mid \X, \R] &= \frac{\mu_{\mathrm{low}}}{  \mu_{u,\I}^{X(\partial B)}(X_u)} =    \left(1 - \frac{1}{5|S_{\ell^*}(v)|}\right) \frac{\mu_{u,\I}^{X(\Psi) \uplus \sigma}(X_u)}{ \mu_{u,\I}^{X(\partial B)}(X_u)}= \left(1 - \frac{1}{5|S_{\ell^*}(v)|}\right) \frac{\mu_{u,\I}^{X(\Psi) \uplus \sigma}(X_u)}{ \mu_{u,\I}^{X(\Psi) \uplus X(S)}(X_u)}\\
&\geq 	\left(1 - \frac{1}{5|S_{\ell^*}(v)|}\right)^2\qquad (\text{by } S \subseteq S_{\ell^*}(u) \text{ and  \Cref{condition-SSM-ratio}})\\
&\geq 1 - \frac{2}{5|S|}. 
  \end{align*}
This proves \Cref{claim-lower-bound}. Besides, we do not need to enumerate all configurations in $[q]^S$ to compute $\mu_{\min}$, the expected running time of \Cref{alg:perfect-sampler-gen} can be reduced to $ O\left( n\cdot q^{ \Delta^{\ell^*} }  \right)$.
\end{remark}

\section{Analysis of strong spatial mixing}
In this section, we use \Cref{theorem-general-ratio} to prove other results mentioned in \Cref{section-results}.
We analyze the strong spatial mixing properties for the classes of spin systems mentioned in \Cref{section-results}, so that we can use~\Cref{theorem-general-ratio} to prove the existences of perfect samplers.
\subsection{Spin systems on sub-exponential neighborhood growth graphs}
In this section, we prove \Cref{theorem-sub-graph} using  \Cref{theorem-general-ratio}.
We need the following proposition, which explains the relation between the multiplicative form of decay in~\eqref{eq:condition-SSM} and the 
additive form of decay in~\eqref{eq:condition-standard-SSM}. Similar results appeared in~\cite{K98, feng2018local,spinka2018finitary,alexander2004mixing}.
\begin{proposition}
\label{proposition-SSM}
Let $\delta: \mathbb{N} \to \mathbb{N}$ be a non-increasing function.
Let $\mathfrak{I}$ be a class of permissive spin system instances exhibiting strong spatial mixing with decay rate $\delta$.
For every instance $\I=(G, [q], \boldsymbol{b}, \boldsymbol{A}) \in \mathfrak{I}$, where $G=(V,E)$,
for every $v \in V$, $\Lambda \subseteq V$, and any two partial configurations $\sigma, \tau \in [q]^\Lambda$ with $\ell \geq 2$,
\begin{align}
\label{eq-implied-SSM}
\forall a \in [q],\quad
\min\left(\left\vert \frac{\mu_{v, \I}^\sigma(a)}{ \mu_{v,\I}^\tau(a) } - 1 \right\vert, 1\right) \leq 10 q\cdot|S_{\lfloor \ell / 2 \rfloor }(v)|\cdot\delta(\lfloor \ell / 2 \rfloor) \quad(\text{with the convention $0/0 = 1$}),
\end{align}
where $\ell \triangleq \min\{\dist_G(v, u) \mid u \in \Lambda,\ \sigma(u)\neq\tau(u)\}\geq 2$ and  $S_{\lfloor \ell / 2 \rfloor}(v) \triangleq \{u \in V \mid \dist_G(v, u) = \lfloor \ell / 2 \rfloor\}$ denotes the sphere of radius $\lfloor \ell / 2 \rfloor$ centered at $v$ in $G$.
\end{proposition}
The proof of \Cref{proposition-SSM} is deferred to the end of this section.
We first use \Cref{proposition-SSM} to prove \Cref{theorem-sub-graph}.
\begin{proof}[Proof of \Cref{theorem-sub-graph}]
Let $q$ be a finite integer.
Suppose  $\mathfrak{I}$ is a class of $q$-state spin systems that is defined on a class of graphs that have sub-exponential neighborhood growth in \Cref{definition-sub-exp-graph} with function $s:\mathbb{N}\to\mathbb{N}$.
Suppose $\mathfrak{I}$ exhibits strong spatial mixing with exponential decay with some constants $\alpha > 0, \beta > 0$. Let $\delta$ be the function $\delta(x) = \alpha \exp(-\beta x)$.

Fix a instance $\I=(G, [q], \boldsymbol{b}, \boldsymbol{A}) \in \mathfrak{I}$, where $G=(V,E)$. 
By \Cref{proposition-SSM},  we have
for any $\Lambda \subseteq V$, any $v \in V$, and any two partial configurations $\sigma, \tau \in [q]^\Lambda$
satisfying $\ell \triangleq \min\{\dist_G(v, u) \mid u \in \Lambda,\ \sigma(u)\neq\tau(u)\}\geq 2$,
\begin{align}
\label{eq-proof-app-1}
\forall a \in [q],\quad
\min\left(\left\vert \frac{\mu_{v, \I}^\sigma(a)}{ \mu_{v,\I}^\tau(a) } - 1 \right\vert, 1\right) &\leq 10 q \cdot |S_{\lfloor \ell / 2 \rfloor }(v)| \cdot \delta(\lfloor \ell / 2 \rfloor ) \notag\\
(\text{by }|S_r(v)| \leq s(r))\quad &\leq 10 q \cdot s(\lfloor \ell / 2 \rfloor) \cdot \delta(\lfloor \ell / 2 \rfloor )
\end{align}
Note that $\delta(\ell) = \alpha \exp(-\beta \ell)$.
We take $\ell_0 = \ell_0(q, \alpha, \beta, s)$ sufficiently large such that $\ell_0 \geq 2$ and
\begin{align}
\label{eq-proof-app-2}
10\alpha q  \cdot s(\lfloor\ell_0 / 2\rfloor )  \exp(-\beta \lfloor\ell_0 / 2\rfloor) \leq \frac{1}{5s( \ell_0 )} \leq \frac{1}{5 \left\vert S_{\ell_0}(v) \right\vert}.
\end{align}
Note that~\eqref{eq-proof-app-2} is equivalent to
\begin{align}
\label{eq-set-ell}
 \alpha\exp(-\beta \lfloor\ell_0 / 2\rfloor) \leq \frac{1}{50q \cdot s(\lfloor\ell_0 / 2\rfloor) \cdot s(\ell_0)}.
\end{align}
Such $\ell_0 = \ell_0(q, \alpha, \beta, s)$ must exist because $s(r)=\exp(o(r))$ for all $r \geq 0$.
Combining~\eqref{eq-proof-app-1} and~\eqref{eq-proof-app-2} implies that $\I$ satisfies \Cref{condition-SSM-ratio} with $\ell_0 \geq 2$.
By \Cref{theorem-general-ratio},  if the parameter $\ell$ in Algorithm~\ref{alg:perfect-sampler-gen} is set so that $\ell = \ell_0-1$, given $\I$, the expected running time of \Cref{alg:perfect-sampler-gen} is
$
n\cdot q^{O\left(  \Delta^{\ell_0}  \right)}.
$
Since $\ell_0 = O(1)$, $q=O(1)$ and $\Delta \leq s(1) = O(1)$,
the expected running time of \Cref{alg:perfect-sampler-gen} is  $O(n)$.
\end{proof}

We now prove \Cref{proposition-SSM}. Similar results are proved in~\cite{K98,feng2018local,spinka2018finitary,alexander2004mixing}.

\begin{proof}[Proof of \Cref{proposition-SSM}]
Fix a instance $\I=(G, [q], \boldsymbol{b}, \boldsymbol{A}) \in \mathfrak{I}$, where $G=(V,E)$.
Fix two partial configurations $\sigma, \tau \in [q]^\Lambda$
with $\ell \triangleq \min\{\dist_G(v, u) \mid u \in \Lambda,\ \sigma(u)\neq\tau(u)\}$.  We use $D \triangleq \{v \in \Lambda \mid \sigma(v) \neq \tau(v) \}$ to denote the set at which $\sigma$ and $\tau$ disagree. Fix a spin $a \in [q]$. Since $\I$ is permissive (\Cref{definition-locally-admissible}),  we have
\begin{align*}
\mu^{\sigma}_{v,\I}(a) = 0 \quad\Longleftrightarrow\quad 	b_v(a)\prod_{u \in \Gamma_G(v) \cap \Lambda}A_{uv}(a,\sigma(v)) = 0,\\
\mu^{\tau}_{v,\I}(a) = 0 \quad\Longleftrightarrow\quad 	b_v(a)\prod_{u \in \Gamma_G(v) \cap \Lambda}A_{uv}(a,\tau(v)) = 0,
\end{align*}  
where $\Gamma_G(v)$ is the neighborhood of $v$ in $G$.
Since $\ell \geq 2$, there is no edge between $v$ and $D$, we have $\Gamma_G(v) \cap D = \emptyset$.
This implies $\mu^{\sigma}_{v,\I}(a)  = 0$ if and only if $\mu^{\tau}_{v,\I}(a)  = 0$. If $\mu^{\sigma}_{v,\I}(a)  = \mu^{\tau}_{v,\I}(a)  = 0$, the proposition holds trivially. We assume
\begin{align}
\label{eq-proof-assume}
\mu^{\sigma}_{v,\I}(a) >0 \land \mu^{\tau}_{v,\I}(a) > 0.
\end{align}
 
Define the set of vertices 
$H \triangleq S_{\lfloor \ell/2 \rfloor}(v) \setminus \Lambda$,	
where $S_{\lfloor \ell /2 \rfloor}(v) = \{u \in V \mid \dist(u,v) = \lfloor \ell/2 \rfloor \}$ is the sphere of radius $\lfloor \ell/2 \rfloor$ centered at $v$ in graph $G$. 
By the definitions, we have $H \cap D = \emptyset$. 
If $H = \emptyset$, then $S_{\lfloor \ell/2 \rfloor}(v)  \subseteq \Lambda$, 
the proposition holds due to the conditional independence property. In the rest of the proof, we assume $H \neq \emptyset$.

For any two disjoint sets $S,S' \subseteq V$ and any partial configurations $\eta \in [q]^S, \eta' \in [q]^{S'}$, we use $\mu_{\I}^{S\gets \eta, S' \gets \eta'}$ to denote the distribution $\mu_{\I}^{\eta \uplus \eta'}$.

For any  $\rho \in [q]^H$ satisfying $\mu^{\Lambda \gets \sigma,v \gets a}_{H,\I}(\rho) > 0$ and $\mu^{\Lambda \gets \tau,v \gets a}_{H,\I}(\rho) > 0$
we have
\begin{align*}
\mu^{\sigma}_{v,\I}(a) = \frac{ \mu_{H,\I}^{\sigma}(\rho) \cdot \mu^{\sigma \uplus \rho}_{v,\I}(a) }{\mu^{\Lambda \gets \sigma,v \gets a}_{H,\I}(\rho)},\quad
 \mu^{ \tau}_{v,\I}(a)= \frac{ \mu_{H,\I}^{\tau}(\rho) \cdot \mu^{\tau \uplus \rho}_{v,\I}(a) }{\mu^{\Lambda \gets \tau,v \gets a}_{H,\I}(\rho)}.
\end{align*}
The first equation holds since $\mu_{H,\I}^{\sigma}(\rho) \cdot \mu^{\sigma \uplus \rho}_{v,\I}(a) =\mu^{\sigma}_{v,\I}(a) \cdot \mu^{\Lambda \gets \sigma,v \gets a}_{H,\I}(\rho)$ and $\mu^{\Lambda \gets \sigma,v \gets a}_{H,\I}(\rho) > 0$; 
the second equation holds similarly. 
Note that $\mu^{\sigma}_{v,\I}(a) >0$ and $\mu^{\tau}_{v,\I}(a) > 0$.
We have
\begin{align*}
\frac{\mu^{\sigma}_{v,\I}(a)}{\mu^{ \tau}_{v,\I}(a)} = \left(\frac{\mu^{ \sigma \uplus \rho}_{v,\I}(a)}{\mu^{\tau \uplus \rho}_{v,\I}(a)}\right)  \left(\frac{\mu_{H,\I}^{ \sigma}(\rho) }{\mu_{H,\I}^{\tau}(\rho) }\right)  \left( \frac{\mu^{\Lambda \gets \tau,v \gets a}_{H,\I}(\rho)}{\mu^{\Lambda \gets \sigma,v \gets a}_{H,\I}(\rho)} \right).
\end{align*}
Note that $(\Lambda \setminus D ) \cup H$ separates $v$ from $D$ in graph $G$, 
and the two configurations $\sigma \uplus \rho$ and $\tau \uplus \rho$ disagree only at $D$. 
By the conditional independence property, we have $\mu^{\sigma \uplus \rho}_{v,\I}(a) = \mu^{\tau \uplus \rho}_{v,\I}(a)$. Hence, we have
\begin{align}
\label{eq-proof-ratios}
\frac{\mu^{\sigma}_{v,\I}(a)}{\mu^{ \tau}_{v,\I}(a)} = \left(\frac{\mu_{H,\I}^{\sigma}(\rho) }{\mu_{H,\I}^{ \tau}(\rho) }\right)  \left( \frac{\mu^{\Lambda \gets \tau,v \gets a}_{H,\I}(\rho)}{\mu^{\Lambda \gets \sigma,v \gets a}_{H,\I}(\rho)} \right).
\end{align}
Note that~\eqref{eq-proof-ratios} holds for any $\rho \in [q]^H$ satisfying $\mu^{\Lambda \gets \sigma,v \gets a}_{H,\I}(\rho) > 0$ and $\mu^{\Lambda \gets \tau,v \gets a}_{H,\I}(\rho) > 0$. Our goal is to choose a suitable $\rho$ and bound the RHS.
Let
\begin{align*}
\epsilon \triangleq \delta(\lfloor \ell / 2 \rfloor). 	
\end{align*}
Without loss of generality,
we assume
\begin{align}
\label{eq-proof-assume2}
10 q\cdot|S_{\lfloor \ell / 2 \rfloor }(v)|\cdot\delta(\lfloor \ell / 2 \rfloor) = 10 q \epsilon\cdot|S_{\lfloor \ell / 2 \rfloor }(v)| < 1.	
\end{align}
If~\eqref{eq-proof-assume2} does not hold, then the inequality~\eqref{eq-implied-SSM} holds trivially.
We have the following claim.
\begin{claim}
\label{claim-exist-rho}
Assume~\eqref{eq-proof-assume2}.
There exists a configuration $\rho \in [q]^H$ satisfying $\mu^{\Lambda \gets \sigma,v \gets a}_{H,\I}(\rho) > 0$ and $\mu^{\Lambda \gets \tau,v \gets a}_{H,\I}(\rho) > 0$ such that
\begin{align*}
\left( 1 - \frac{2q\epsilon}{1+q\epsilon} \right)^{2m} \leq \left(\frac{\mu_{H,\I}^{\sigma}(\rho) }{\mu_{H,\I}^{ \tau}(\rho) }\right) \left( \frac{\mu^{\Lambda \gets \tau,v \gets a}_{H,\I}(\rho)}{\mu^{\Lambda \gets \sigma,v \gets a}_{H,\I}(\rho)} \right) \leq \left( 1 + \frac{2q\epsilon}{1-q\epsilon} \right)^{2m},
\end{align*}
where $ m \triangleq |S_{\lfloor \ell / 2 \rfloor }(v)|$ and $\epsilon \triangleq \delta(\lfloor \ell / 2 \rfloor)$.
\end{claim}

The inequality \eqref{eq-proof-assume2} implies that
\begin{align}
  q\epsilon m\le \frac{1}{10}.
  \label{eq-proof-assume2p}
\end{align}
Combining~\Cref{claim-exist-rho} with the above, we have
\begin{align*}
\left(\frac{\mu_{H,\I}^{\sigma}(\rho) }{\mu_{H,\I}^{ \tau}(\rho) }\right)  \left( \frac{\mu^{\Lambda \gets \tau,v \gets a}_{H,\I}(\rho)}{\mu^{\Lambda \gets \sigma,v \gets a}_{H,\I}(\rho)} \right) &\leq \exp \left( \frac{4q\epsilon m}{1-q\epsilon} \right)\\
 &\leq \exp\left( 5q\epsilon m \right) \tag{by~\eqref{eq-proof-assume2p}}\\
 &\leq 1 + 10q\epsilon m. \tag{by~\eqref{eq-proof-assume2p}}
\end{align*}
Similarly, we have
\begin{align*}
\left(\frac{\mu_{H,\I}^{\sigma}(\rho) }{\mu_{H,\I}^{ \tau}(\rho) }\right) \left( \frac{\mu^{\Lambda \gets \tau,v \gets a}_{H,\I}(\rho)}{\mu^{\Lambda \gets \sigma,v \gets a}_{H,\I}(\rho)} \right) &\geq \left( 1 - 2q\epsilon \right)^{ 2m}\\
&\geq \exp\left( -8q\epsilon m \right) \tag{by~\eqref{eq-proof-assume2p}}\\
 &\geq 1 - 10q\epsilon m.
\end{align*}
Recall $ m \triangleq |S_{\lfloor \ell / 2 \rfloor }(v)|$ and $\epsilon \triangleq \delta(\lfloor \ell / 2 \rfloor)$. This proves the proposition.
\end{proof}
\begin{proof}[Proof of Claim~\ref{claim-exist-rho}]
Suppose $|H| = h\geq1$.
Let $H = \{v_1,v_2,\ldots,v_h\}$. Define a sequence of subsets $H_0,H_1,\ldots,H_h$ as $H_i \triangleq \{v_j \mid 1\leq j\leq i\}$. Note that $H_0 = \emptyset$ and $H_h=H$. We now construct the configuration $\rho \in [q]^H$ by the following $h$ steps.
\begin{itemize}
\item initially, $\rho  = \emptyset$ is an empty configuration;
\item in $i$-th step, note that $\rho \in [q]^{H_{i-1}}$, choose $c_i \in [q]$ that  maximizes $\mu_{v_i,\I}^{\Lambda\gets \sigma,v \gets a,H_{i-1} \gets \rho}(c_i)$ (break tie arbitrarily), extend $\rho$ further at position $v_i$ and set $\rho(v_i) = c_i$, thus $\rho \in [q]^{H_{i}}$ after the $i$-th step.
\end{itemize}
By the construction, we have
\begin{align*}
\forall  1\leq i \leq  h,\quad 	\mu_{v_i,\I}^{\Lambda\gets \sigma,v \gets a,H_{i-1} \gets \rho(H_{i-1})}(\rho(v_i)) \geq \frac{1}{q} > 0.
\end{align*}
Recall $H \triangleq S_{\lfloor \ell/2 \rfloor}(v) \setminus \Lambda$. We have $\dist_G(H, D) \geq \ell - \lfloor \ell / 2 \rfloor \geq \lfloor\ell / 2\rfloor$, where $D$ is the set at which $\sigma$ and $\tau$ disagree, and $\dist_G(H, D) \triangleq \min\{ \dist_G(u_1,u_2) \mid u_1 \in H \land u_2 \in D \}$.
Recall $\epsilon \triangleq \delta(\lfloor \ell / 2 \rfloor)$ and $\delta$ is a non-increasing function.
By the strong spatial mixing property in~\Cref{definition-standard-SSM}, we have
\begin{align*}
\forall  1\leq i \leq  h,\quad 	\mu_{v_i,\I}^{\Lambda\gets \tau,v \gets a,H_{i-1} \gets \rho(H_{i-1})}(\rho(v_i)) \geq \frac{1}{q} -\epsilon > 0,
\end{align*}
where $\frac{1}{q} -\epsilon > 0$ holds due to~\eqref{eq-proof-assume2}. By the chain rule, we have $\mu_{H,\I}^{\Lambda \gets \sigma, v \gets a}(\rho) > 0$ and $\mu_{H,\I}^{\Lambda \gets \tau, v \gets a}(\rho) > 0$. 

We now prove that $\rho$ satisfies the inequalities in \Cref{claim-exist-rho}.
For any $1\leq i \leq h$, define
\begin{align}
\label{eq-def-pi}
p_i \triangleq 		\mu_{v_i,\I}^{\Lambda\gets \sigma,v \gets a,H_{i-1} \gets \rho(H_{i-1})}(\rho(v_i)).
\end{align}
Recall that $\dist_G(H,v) \geq \lfloor\ell / 2\rfloor$ and $\dist_G(H, D) \geq \lfloor\ell / 2\rfloor$. Recall $\epsilon \triangleq \delta(\lfloor \ell / 2 \rfloor)$ and $\delta$ is a non-increasing function. By the strong spatial mixing property in~\Cref{definition-standard-SSM}, we have for any $c \in [q]$ and any $1\leq i \leq  h$,
\begin{align}
0 <p_i -\epsilon \leq \mu_{v_i,\I}^{\Lambda\gets \sigma,v \gets c,H_{i-1} \gets \rho(H_{i-1})}(\rho(v_i)) \leq p_i + \epsilon,\label{eq-lu-1}\\
0< p_i -\epsilon \leq \mu_{v_i,\I}^{\Lambda\gets \tau,v \gets c,H_{i-1} \gets \rho(H_{i-1})}(\rho(v_i)) \leq p_i + \epsilon.\label{eq-lu-2}
\end{align}
Note that $p_i - \epsilon \geq \frac{1}{q} - \epsilon > 0$ due to~\eqref{eq-proof-assume2}. Combining ~\eqref{eq-lu-1}, ~\eqref{eq-lu-2} and the chain rule implies
\begin{align*}
\forall c,c'\in[q],\quad
\prod_{i=1}^h\left(\frac{p_i - \epsilon}{p_i + \epsilon}\right)	 \leq \frac{\mu^{\Lambda \gets \tau,v \gets c}_{H,\I}(\rho)}{\mu^{\Lambda \gets \sigma,v \gets c'}_{H,\I}(\rho)}  \leq \prod_{i=1}^h \left(\frac{p_i + \epsilon}{p_i- \epsilon}\right)
\end{align*}
Note that $p_i \geq \frac{1}{q}$ for all $1\leq i \leq h$ due the construction of $\rho$, and $q\epsilon < 1$ due to~\eqref{eq-proof-assume2}.  We have
\begin{align}
\label{eq-proof-lower-up-1}
\forall c,c'\in[q],\quad
\left( 1 - \frac{2q\epsilon}{1+q\epsilon} \right)^h \leq \frac{\mu^{\Lambda \gets \tau,v \gets c}_{H,\I}(\rho)}{\mu^{\Lambda \gets \sigma,v \gets c'}_{H,\I}(\rho)} \leq \left( 1 + \frac{2q\epsilon}{1-q\epsilon} \right)^h,
\end{align}
and
\begin{align}
\label{eq-proof-lower-up-2}
\forall c,c'\in[q],\quad
\left( 1 - \frac{2q\epsilon}{1+q\epsilon} \right)^h \leq \frac{\mu^{\Lambda \gets \sigma,v \gets c}_{H,\I}(\rho)}{\mu^{\Lambda \gets \tau,v \gets c'}_{H,\I}(\rho)} \leq \left( 1 + \frac{2q\epsilon}{1-q\epsilon} \right)^h.
\end{align}
Note that
\begin{align*}
\mu^{\sigma}_{H,\I}(\rho) = \sum_{c \in [q]}	\mu^{\sigma}_{v,\I}(c) \mu^{\Lambda \gets \sigma,v \gets c}_{H,\I}(\rho),\quad \mu^{\tau}_{H,\I}(\rho) = \sum_{c \in [q]}	\mu^{\tau}_{v,\I}(c) \mu^{\Lambda \gets \tau,v \gets c}_{H,\I}(\rho).
\end{align*}
$\mu^{\sigma}_{H,\I}(\rho)$ is a convex combination of $\mu^{\Lambda \gets \sigma,v \gets c}_{H,\I}(\rho)$, and $\mu^{\tau}_{H,\I}(\rho)$ is a convex combination of $\mu^{\Lambda \gets \tau,v \gets c}_{H,\I}(\rho)$. By~\eqref{eq-proof-lower-up-1} and~\eqref{eq-proof-lower-up-2}, it holds that
\begin{align*}
\left( 1 - \frac{2q\epsilon}{1+q\epsilon} \right)^{2h } \leq \left(\frac{\mu_{H,\I}^{\sigma}(\rho) }{\mu_{H,\I}^{ \tau}(\rho) }\right) \left( \frac{\mu^{\Lambda \gets \tau,v \gets a}_{H,\I}(\rho)}{\mu^{\Lambda \gets \sigma,v \gets a}_{H,\I}(\rho)} \right) \leq \left( 1 + \frac{2q\epsilon}{1-q\epsilon} \right)^{2h }.
\end{align*}
Note that $h = |H|$ and $H \subseteq S_{\lfloor \ell /2 \rfloor}(v)$, then $m = |S_{\lfloor \ell /2 \rfloor}(v)| \geq h$.
This proves the claim.
\end{proof}

\subsection{Spin systems on general graphs}
In this section, we prove \Cref{theorem-general} by showing that \Cref{condition-lower-bound} implies \Cref{condition-SSM-ratio}.
\begin{proof}[Proof of \Cref{theorem-general}]
Fix a instance $\I = (G,[q], \boldsymbol{A},\boldsymbol{b}) \in \mathfrak{I}$
satisfying \Cref{condition-lower-bound} with parameter $\ell = \ell(q) \geq  2$.
Fix subset $\Lambda \subseteq V$ and vertex $v \in V \setminus \Lambda$.
For any two partial configurations $\sigma, \tau \in [q]^\Lambda$ satisfying $\min\{\dist_G(v, u) \mid u \in \Lambda,\ \sigma(u)\neq\tau(u)\} =\ell \geq 2$, we claim
\begin{align}
\label{eq-claim-application}
\forall a \in [q], \quad \mu^{\sigma}_{v,\I}(a) = 0 \quad \Longleftrightarrow \quad \mu^{\tau}_{v,\I}(a) = 0.
\end{align}
Let $D \triangleq \{u \in \Lambda \mid \sigma(u) \neq \tau(u) \}$, $H\triangleq \Lambda \setminus D$ and $\rho \triangleq \sigma_H = \tau_H$. 
Since $\ell \geq 2$, $\Gamma_G(v) \cap \Lambda = \Gamma_G(v) \cap H$, where $\Gamma_G(v)$ is the neighborhood of $v$ in $G$.
Since $\I$ is a permissive spin system (\Cref{definition-locally-admissible}), $\mu^{\sigma}_{v,\I}(a) = 0$ if and only if $b_v(a)\prod_{u \in \Gamma_G(v) \cap H}A_{uv}(a,\rho_u) = 0$; similarly, $\mu^{\tau}_{v,\I}(a) = 0$ if and only if $b_v(a)\prod_{u \in \Gamma_G(v) \cap H}A_{uv}(a,\rho_u) = 0$. This proves~\eqref{eq-claim-application}. 

If $\mu^{\sigma}_{v,\I}(a) = \mu^{\tau}_{v,\I}(a) = 0$, then~\eqref{eq:condition-SSM} holds trivially. Otherwise, by \Cref{condition-lower-bound}, $\mu^{\sigma}_{v,\I}(a) \geq \gamma$ and $\mu^{\tau}_{v,\I}(a) \geq \gamma$, where $\gamma = \gamma(\Lambda, v)> 0$ is positive and depends only on $\Lambda$ and $v$.
By~\eqref{eq:condition-stronger-lower} and~\eqref{eq:condition-stronger-SSM-general}, we have
\begin{align*}
\left\vert \frac{\mu_{v, \I}^\sigma(a)}{ \mu_{v,\I}^\tau(a) } - 1 \right\vert \leq \frac{\gamma+ \DTV{\mu^{\sigma}_{v,\I}}{\mu^{\tau}_{v,\I}} }{\gamma} - 1 \leq \frac{1}{5|S_{\ell}(v)|}.
\end{align*}
This implies that any instance   $\I = (G,[q], \boldsymbol{A},\boldsymbol{b}) \in \mathfrak{I}$ satisfies \Cref{condition-SSM-ratio} with parameter $\ell=\ell(q) \geq 2$.
\Cref{theorem-general} is a corollary of \Cref{theorem-general-ratio}.
\end{proof}


\subsection{Uniform list coloring}
\label{section-list-coloring}
We now prove  \Cref{corollary-main-coloring}.
Let $\mathfrak{L}$ be a class of list coloring instances with at most $q$ colors for a finite $q>0$.
Let $\alpha^* \approx 1.763\ldots$ be the positive root of the equation $x^x =\mathrm{e}$.
Suppose there exist  $\alpha > \alpha^*$ and $\beta \geq \frac{\sqrt{2}}{\sqrt{2}-1}$ satisfying $(1-1/\beta) \alpha \mathrm{e}^{\frac{1}{\alpha}(1-1/\beta)} > 1$ such that for all $\I=(G=(V,E),[q],\+L) \in \mathfrak{L}$, the graph $G$ is triangle-free and
\begin{align*}
\forall v \in V, \quad |L(v)| \geq \alpha \deg_G(v) + \beta.
\end{align*}
Gamarnik, Katz, and Misra~\cite{GKM15} proved that $\mathfrak{L}$ exhibits the strong spatial mixing with exponential decay. If $\mathfrak{L}$ is defined on sub-exponential neighborhood growth graphs, then by \Cref{theorem-sub-graph}, the linear time perfect sampler exists for every instance in $\mathfrak{L}$.

There are two remaining cases in \Cref{corollary-main-coloring}.
We now assume that  the class $\mathfrak{L}$ of list coloring instances satisfies one of the following two conditions.
\begin{enumerate}[label=(\Roman*)]
\item \label{condition-color-1} 
there is an $s:\mathbb{N}\to \mathbb{N}$ with $s(\ell) = \exp(o(\ell))$ such that
for any $\I=(G = (V,E),[q],\+L) \in \mathfrak{L}$,
\begin{align*}
\forall v \in V, \ell \geq 0, \quad & |S_{\ell}(v)|\leq s(\ell),  \\
\forall v \in V, \quad &|L(v)| \geq 2\deg_G(v);	
\end{align*}
\item \label{condition-color-2} for any $\I=(G=(V,E),[q],\+L) \in \mathfrak{L}$,
\begin{align*}
\forall v \in V, \quad |L(v)| \geq \Delta^2  - \Delta + 2.
\end{align*}
\end{enumerate}


%
%
\begin{lemma}
\label{lemma-coloring-to-real}
Let $\mathfrak{L}$ be a class of list coloring instances with at most $q$ colors for a finite $q>0$. Suppose $\mathfrak{L}$ satisfies \ref{condition-color-1} or \ref{condition-color-2}. There exist finite $A > 0$ and $\theta > 0$ such that for every  $\I =(G,[q],\List) \in \mathfrak{L}$, where $G=(V,E)$, 
%
for any $v \in V$, any $\Lambda \subseteq V$, and any $\sigma, \tau \in [q]^\Lambda$
with $\ell \triangleq \min\{\dist_G(v, u) \mid u \in \Lambda,\ \sigma(u)\neq\tau(u)\} = \Omega(q\log q)$, it holds that
\begin{align*}
  \forall a \in [q]: \quad \left\vert \frac{\mu_{v, \I}^\sigma(a)}{ \mu_{v,\I}^\tau(a) } - 1 \right\vert \leq \frac{ A \mathrm{e}^{-\theta \ell } }{\left\vert S_{\ell}(v) \right\vert}
  \quad(\text{with the convention $0/0 = 1$}),
\end{align*}
where $A = A(q, s) > 0$ and $\theta= \frac{1}{2q}> 0$ if $\mathfrak{L}$ satisfies~\ref{condition-color-1} with the function $s: \mathbb{N}\rightarrow\mathbb{N}$;  
or $A = \mathrm{poly}(q)$ and $\theta = \frac{1}{2q^2} > 0$ if $\mathfrak{L}$ satisfies~\ref{condition-color-2}.
\end{lemma}


\Cref{theorem-general-ratio} together with Lemma~\ref{lemma-coloring-to-real} proves the remaining two cases in \Cref{corollary-main-coloring}.
We take a sufficiently large $\ell^*$ such that $\ell^* = \Omega(q \log q)$ and $A\mathrm{e}^{-\theta \ell^*} \leq \frac{1}{5}$. 
By Lemma~\ref{lemma-coloring-to-real}, instances of $\mathfrak{L}$ satisfy Condition~\ref{condition-SSM-ratio} with this $\ell^* \geq 2$.
Thus the perfect sampler exists due to \Cref{theorem-general-ratio}.
Note that $\ell^*$ depends only on $q$ and the function $s$, and for any instance  $\I \in \mathfrak{L}$, the maximum degree $\Delta \leq q$. 
Thus, the expected running time of our algorithm is $n \cdot q^{O(q^{\ell^*})} = O(n)$.
Furthermore, if $\mathfrak{L}$ satisfies~\ref{condition-color-2}, then $\ell^* = \Theta(q^2\log q)$, thus the expected running time is $n \cdot \exp(\exp(\mathrm{poly}(q)))$.


\subsubsection{The multiplicative SSM of list coloring (proof of \texorpdfstring{\Cref{lemma-coloring-to-real}}{\texttwoinferior})}
In \cite[Theorem~3]{GKM15},  Gamarnik, Katz, and Misra established the best known strong spatial mixing result for list colorings in bounded degree graphs.
This is almost what we need, except that we want to control the decay rate under conditions~\ref{condition-color-1} and ~\ref{condition-color-2}.
Going through the proof of \cite[Theorem~3]{GKM15} and keeping track of the decay rate, 
we have the proposition below.
The similar analysis technique are also used in~\cite{liu2019deterministic}.

\begin{proposition}[\cite{GKM15}]
\label{proposition-single-instance}
Let $\I=(G, [q], \List)$ be a list coloring instance, where $G=(V, E)$. Assume that $\I$ satisfies $|L(v)| \geq \deg_G(v)+ 1$ for all $v \in V$. Suppose  
\begin{align*}
\max_{u \in V} \frac{\deg_G(u) - 1 }{|L(u)| - \deg_G(u)} \leq \chi < 1.   
\end{align*}
Then for any  $\Lambda \subseteq V$, any vertex $v \in V \setminus \Lambda$, and any two partial colorings $\sigma, \tau \in [q]^\Lambda$
satisfying $\ell \triangleq \min\{\dist_G(v, u) \mid u \in \Lambda,\ \sigma(u)\neq\tau(u)\} = \Omega(\frac{\log q}{\log(1/\chi)})$, it holds that
\begin{align*}
\forall a \in [q]:\quad
\left\vert \frac{\mu^\sigma_{v,\I}(a)}{\mu^\tau_{v, \I}(a)} - 1 \right\vert 	\leq B\chi^{\ell}, \quad\text{(with convention $0/0 = 1$)}
\end{align*}
where $B = \mathrm{poly}(q/\chi)$ depends only on $q$ and $\chi$.
\end{proposition}

\begin{proof}[Proof of Lemma~\ref{lemma-coloring-to-real}]
Fix a instance $\I = (G, [q],\List) \in \mathfrak{L}$, where $G=(V,E)$.
Suppose $\mathfrak{L}$ satisfies Condition in~\ref{condition-color-1}.
We have
\begin{align*}
\max_{u \in V} \frac{\deg_G(u) - 1 }{|L(u)| - \deg_G(u)} \leq \max_{u \in V} \frac{\deg_G(u) - 1 }{ \deg_G(u)} = \frac{\Delta-1}{\Delta} \leq \frac{q-1}{q}.
\end{align*}
The $\chi$ and $B$ in \Cref{proposition-single-instance} can be set as $\chi = \frac{q-1}{q}$ and $B = \mathrm{poly}(q/\chi) \leq B_{\max} = \mathrm{poly}(q)$.
Then
for any subset $\Lambda \subseteq V$, any vertex $v \in V \setminus \Lambda$, any two colorings $\sigma,\tau \in [q]^\Lambda$ that disagree on $D \subseteq \Lambda$ satisfying $\ell \triangleq \min\{\dist_G(u,v)\mid u \in D\} = \Omega(\frac{\log q}{\log 1/\chi}) = \Omega(q\log q)$, it holds that
\begin{align*}
  \forall a \in [q]:\quad
  \left\vert \frac{\mu^\sigma_{v,\I}(a)}{\mu_{v,\I}^\tau(a)}-1   \right\vert	
  \leq B_{\max}\chi^{\ell}
  \leq B_{\max}\cdot\frac{ \left\vert S_\ell(v)\right\vert}{\left\vert S_\ell(v)\right\vert}\cdot\chi^{\ell}.
\end{align*}
Since $G$ has sub-exponential growth, 
we have that $\left\vert S_\ell(v)\right\vert \leq s(\ell) = \exp(o(\ell))$.
Thus,
\begin{align*}
  \forall a \in [q]:\quad
  \left\vert \frac{\mu^\sigma_{v,\I}(a)}{\mu_{v,\I}^\tau(a)}-1   \right\vert 
  \leq B_{\max}\cdot\frac{s(\ell) }{\left\vert S_\ell(v)\right\vert}\cdot\left(\frac{q-1}{q} \right)^{\ell} 
  \leq \frac{s(\ell)B_{\max} }{\left\vert S_\ell(v)\right\vert}\cdot\mathrm{e}^{-\ell/q}
  \leq \frac{ A\mathrm{e}^{-\theta \ell}}{\left\vert S_\ell(v)\right\vert},
\end{align*}
for some $A=A(q, s) > 0$ and $\theta =\frac{1}{2q} > 0$.

Suppose $\mathfrak{L}$ satisfies~\ref{condition-color-2}.
Recall that $\Delta$ is the maximum degree of graph $G$.
we have
\begin{align*}
\max_{u \in V} \frac{\deg_G(u) - 1 }{|L(u)| - \deg_G(u)} \leq \frac{\Delta-1}{(\Delta-1)^2 + 1}.
\end{align*}
The $\chi$ and $B$ in \Cref{proposition-single-instance} can be set as $\chi = \frac{\Delta-1}{(\Delta-1)^2 + 1}$ and $B = \mathrm{poly}(q/\chi)$. Thus $1/\chi \leq \Delta^2 \leq q^2$. We have $B = \mathrm{poly}(q/\chi) \leq  B_{\max}=\mathrm{poly}(q)$.
For any subset $\Lambda \subseteq V$, any vertex $v \in V \setminus \Lambda$, any two colorings $\sigma,\tau \in [q]^\Lambda$ that disagree on $D \subseteq \Lambda$ satisfying $\ell \triangleq \min\{\dist_G(u,v)\mid u \in D\} = \Omega(\frac{\log q}{\log 1/\chi}) = \Omega(\log q)$, it holds that
\begin{align*}
  \forall a \in [q]:\quad
  \left\vert \frac{\mu^\sigma_{v,\I}(a)}{\mu_{v,\I}^\tau(a)}-1   \right\vert	
  \leq B_{\max}\chi^{\ell}
  \leq B_{\max}\cdot\frac{ \Delta (\Delta - 1)^{\ell-1}}{\left\vert S_\ell(v)\right\vert}\cdot\chi ^{\ell},
\end{align*}
where the last inequality due to $|S_{\ell}(v)| \leq \Delta (\Delta - 1)^{\ell-1}$.
Since  $\chi = \frac{\Delta-1}{(\Delta-1)^2 + 1}$, we have
\begin{align*}
  \forall a \in [q]:\quad
  \left\vert \frac{\mu^\sigma_{v,\I}(a)}{\mu_{v,\I}^\tau(a)}-1   \right\vert 
  &\leq \frac{B_{\max}\Delta}{\Delta-1}\cdot\frac{1}{\left\vert S_\ell(v)\right\vert} \cdot\left(\frac{(\Delta-1)^2}{(\Delta-1)^2 + 1} \right)^{\ell}\\
 (\text{by } \Delta \leq q)\quad&\leq\frac{2B_{\max}}{\left\vert S_\ell(v)\right\vert} \cdot\left(\frac{(q-1)^2}{(q-1)^2 + 1} \right)^{\ell}\\
  &\leq \frac{2B_{\max}}{\left\vert S_\ell(v)\right\vert} \cdot \mathrm{e}^{-\frac{\ell}{2q^2}} = \frac{A\mathrm{e}^{-\theta \ell}}{\left\vert S_\ell(v)\right\vert},
\end{align*}
where $A = 2B_{\max} = \mathrm{poly}(q)$ and $\theta = \frac{1}{2q^2} > 0$.
\end{proof}

\subsection{The monomer-dimer model}
We now prove \Cref{corollary-main-matching}.
We first present the monomer-dimer model instance as a spin system instance, then we use \Cref{theorem-sub-graph} to prove~\Cref{corollary-main-matching}.

Given a graph $G = (V, E)$, we use $G^{*}=(V^*,E^*) = \mathrm{Lin}(G)$ to denote the \emph{line graph} of $G$. Each vertex $v_e \in V^*$ in line graph $G^*$ represents an  edge $e \in E$ in the original graph $G$, and two vertices $v_e,v_{e'}$ in $G^*$ are adjacent if and only if $e$ and $e'$ share a vertex in $G$. We call $S \subseteq V^*$ an \emph{independent set} in $G^*$ if no two vertices in $S$ are adjacent in $G^*$. It is easy to verify that there is a one-to-one correspondence between the matchings in $G$ and the independent sets in $G^*$.

Given a monomer-dimer model instance $\I = (G, \lambda)$, we define a \emph{hardcore model} instance $\I^*=(G^*,\lambda)$ in the line graph $G^* = \mathrm{Lin}(G)$.  Each independent set $S$ in $G^*$ is assigned a weight $w_{\I^*}(S) = \lambda^{|S|}$. Let $\mu_{\I^*}$ be a distribution over all independent sets in $G^*$ such that $\mu_{\I^*}(S) \propto w_{\I^*}(S)$. Hence, $\I^*$ is a spin system instance and $\I^*$ is permissive.
Besides, if we can sample independent sets from $\mu_{\I^*}$, then we can sample matchings from $\mu_{\I}$.

Suppose the class of monomer-dimer model instances $\mathfrak{M}$ satisfies the condition in \Cref{corollary-main-matching}.  
Then, there exist a constant $C$ and a function $s:\mathbb{N}\to \mathbb{N}$ with $s(\ell) = \exp(o(\ell))$ such that for all $\I = (G,\lambda) \in \mathfrak{M}$, $\lambda \leq C = O(1)$, $|S_\ell(v)| \leq s(\ell) = \exp(o(\ell))$ for all $v \in V$ and $\ell \geq 0$, and $\Delta_G \leq s(1) = O(1)$.
Thus, $\mathfrak{M}$ exhibits strong spatial mixing with exponential decay with constants $\alpha = \alpha(C,s)>0$ and $\beta = \beta(C,s)>0$~\cite{bayati2007simple,song2016counting}.
Observe that if $e_1,e_2,\ldots,e_\ell$ is a path of edges in $G$, then $v_{e_1},v_{e_2},\ldots,v_{e_\ell}$ is a path of vertices in $G^*$, and vice versa.
Hence, the following results hold for the class of hardcore instances $\mathfrak{H} = \{\I^*=(G^*,\lambda) \mid \I \in \mathfrak{M}\}$.
\begin{itemize}
\item  The class of hardcore instances $\mathfrak{H}$  exhibits strong spatial mixing with exponential decay with constants $\alpha' = \alpha'(C,s)>0$ and $\beta = \beta(C,s)>0$.
\item for any instance $(G^*,\lambda) \in \mathfrak{H}$, the graph $G^*$  has sub-exponential growth. Suppose $G^* = (V^*,E^*)$ is the line graph of $G=(V,E)$. For all $ e = \{u,v\} \in E$, $\ell \geq 1$, it holds that $|S^*_{\ell}(v_e)| \leq \Delta_G(|S_{\ell - 1}(u)| + |S_{\ell - 1}(v)|) \leq 2s(1)s(\ell - 1) = \exp(o(\ell))$, where  $v_e \in V^*$ represents the edge $e$, $S^*_{\ell}(v_e)$ is the sphere of radius $\ell$ centered at $v_e$ in $G^*$, $S_{\ell-1}(v)$ is the sphere of radius $\ell-1$ centered at $v$ in $G$.
\end{itemize}
Note that the number of vertices in $G^*$ is at most $n\Delta_G = O(n)$, where $n$ is number of vertices in $G$.
\Cref{corollary-main-matching} is a corollary of \Cref{theorem-sub-graph}.

\section{Dynamic sampling}
\label{section-dynamic}
In this section, we use our algorithm to solve the dynamic sampling problem~\cite{FVY19,feng2019dynamicMCMC}.
 In this problem, the Gibbs distribution itself changes dynamically and the algorithm needs to maintain a random sample efficiently with respect to the current Gibbs distribution.
 
We first define the update for the spin system instance.
Let $\I=(G,[q],\boldsymbol{b},\boldsymbol{A})$ be a spin system instance, where $G = (V, E)$.
\begin{itemize}
\item updates for vertices: modifying the vector $b_v$ of vertex $v \in V$;
\item updates for edges: modifying the matrix $A_e$ of edge $e \in E$; or adding new edge $e \notin E$.	
\end{itemize}
We use  $(D_V, D_E, \+C)$ to denote the update for instance $\I$, where $D_V \subseteq V$, $D_E \subseteq \{\{u,v\}\mid u,v \in V, u\neq v\}$, and $\+C = (b_v)_{v \in D_V} \cup (A_e)_{e \in D_E}$. For each $v \in D_V$, we modify its vector to $b_v \in \+C$, and for each $e \in D_E$, we either add the new edge $e$ with matrix $A_e \in \+C$ (if $e \notin E$), or modify its matrix to $A_e \in \+C$ (if $e \in E$).

\begin{definition}[dynamic sampling problem]
\label{definition-dynamic}
Given a spin system instance $\I=(G,[q],\boldsymbol{b},\boldsymbol{A})$ where $G=(V,E)$, a random sample $\X \in [q]^V$ such that $\X \sim \mu_{\I}$, and an update  $(D_V, D_E, \+C)$ that modifies the instance $\I$ to an updated instance $\I'= (G',[q],\boldsymbol{b}',\boldsymbol{A}')$ where $G'=(V, E')$, the algorithm updates $\X$ to a new sample $\X' \in [q]^V$ such that $\X' \sim \mu_{\I'}$.
\end{definition}

\begin{theorem}
\label{theorem-dynamic}
Let $\mathfrak{I}$ be a class of permissive spin systems satisfying Condition~\ref{condition-SSM-ratio}.
There exists an algorithm such that if the updated instance $\I'= (G',[q],\boldsymbol{b}',\boldsymbol{A}') \in \mathfrak{I}$, then the algorithm  solves the dynamic sampling problem  within $\Delta(|D_V| + |D_E|) q^{O(\Delta^{\ell})}$ time in expectation, 
where $\Delta$ is the maximum degree of $G'$ and $\ell = \ell(q) \geq 2$ is determined by \Cref{condition-SSM-ratio}.
\end{theorem}

Suppose $q,\Delta,\ell = O(1)$. By \Cref{theorem-dynamic},
the running time of our algorithm is linear in the size of the update. 
Hence, the efficient dynamic sampling algorithm exists if  strong spatial mixing holds with a rate faster than the neighborhood growth.  
The relation between the spatial mixing property and the static sampling is well studied, we extend such relation further to the dynamic setting.

The dynamic sampling algorithm is given in~\Cref{dynamic-perfect-sampler}. 
\begin{algorithm}[ht]
\SetKwInOut{Input}{Input}
\Input{a spin system instance $\I=(G=(V,E), [q], \boldsymbol{b}, \boldsymbol{A})$, a random sample $\X \sim \mu_{\I}$, and an update  $(D_V, D_E, \+C)$ that modifies $\I$ to $\I' = (G'=(V,E'), [q], \boldsymbol{b}', \boldsymbol{A}')$.}
$\+D \gets D_V \cup \left( \bigcup_{e \in D_E} e\right)$ and $\partial \+D \gets \{v \in V \setminus \+D \mid \exists u \in \+D \text{ s.t. } \{u,v\}\in E' \}$\;
based on $X_{\partial \D}$, modify the partial configuration $X_{\D}$ so that $w_{\I'}(\X) > 0$\label{line-fix}\;
$\R \gets \+D \cup \partial \+D$ \;
\While{$\R \neq \emptyset$}{
	$(\X, \R) \gets \ReSample(\I', \X, \R)$\label{line-dynamic-resample}\; 
}
\Return{$\X$\;}
\caption{Dynamic perfect Gibbs sampler}\label{dynamic-perfect-sampler}
\end{algorithm}

\begin{algorithm}[htbp]
\SetKwInOut{Input}{Input}
\SetKwIF{withprob}{}{}{with probability}{do}{}{}{}
\Input{a spin system instance $\I=(G=(V,E), [q], \boldsymbol{b}, \boldsymbol{A})$, a configuration $\X \in [q]^V$, a non-empty subset $\+R \subseteq V$, and an integer parameter $\ell \geq 0$;}
pick a $u\in\R$ uniformly at random and let $B\gets (B_{\ell}(u)\setminus \R)\cup\{u\}$\;
  let $\mu_{\min}$ be the minimum value of $\mu_{u,\I}^{\sigma}(X_u)$ over all $\sigma\in[q]^{\partial B}$ that $\sigma_{\R\cap\partial B}=X_{\R\cap\partial B}$\;
\withprob{$\frac{\mu_{\min}}{\mu_{u,\I}^{X_{\partial B}}(X_u)}$}{

update $\X$ by redrawing $X_B\sim\mu_{B,\I}^{X_{\partial B}}$\;
	 $\mathcal{R} \gets\mathcal{R} \setminus \{u\}$\;
}
\Else{
	$\mathcal{R} \gets\mathcal{R} \cup \partial B$\;
}
\Return{$(\X,\R)$}
\caption{$\ReSample(\I, \X, \R)$}\label{alg:dynamic-resample}
\end{algorithm}

In \Cref{dynamic-perfect-sampler},
the set $\+D \subseteq V$ contains all the vertices incident to the update.
Note that the input $\X$ can be an infeasible configuration for $\I'$, i.e.~$w_{\I'}(\X) = 0$, because the configuration $X_{\D}$ may violate the new constraints in $\+C$. Hence, in Line~\ref{line-fix}, we modify the configuration $X_{\+D}$ so that $w_{\I'}(\X) > 0$. Given the $X_{\partial \+D}$, this step can be achieved by a simple greedy algorithm since $\I'$ is permissive. Then, we construct the initial $\+R$ as $\+D \cup \partial \+D$. In Line~\ref{line-dynamic-resample}, we call the subroutine $\ReSample$ on the updated instance $\I'$. 

Note that $\+R = \+D \cup \partial\+D$ and $\partial \+D$ separates $\+D$ from $V \setminus \R$ in both $G$ and $G'$.
In Line~\ref{line-fix}, we only modify the partial configuration $X_{\+D}$. 
Such modification only reviews the information of $\X$ in $\+D \cup \partial \+D$. 
Thus, after the modification, the $X_{V \setminus \+R}$ follows the distribution $\mu_{\I}^{X_{\+R}} = \mu_{V \setminus \+R,\I}^{X_{\partial \+D}}$ due to the conditional independence property. Since two instances $\I$ and $\I'$ differ only at the subset $\+D$, due to the conditional independence property, two distributions $\mu_{\I}^{X_{\+R}}$ and $\mu_{\I'}^{X_{\+R}}$ are identical. Thus, the initial $\X, \+R$ satisfies \Cref{condition-invariant} with respect to $\I'$, and $\X$ is a feasible configuration for $\I'$. In each iteration of the \while{} loop, we call the subroutine $\ReSample$ on $\I'$. 
	By the identical proof in \Cref{section-correctness}, the output $\X \sim \mu_{\I'}$.

Let $\Delta$ denote the maximum degree of graph $G'$.
Note that $|\+D| = O(|D_V| + |D_E|)$. The time complexity of the first three lines of \Cref{dynamic-perfect-sampler} is $O(\Delta |\+D|)$.  Note that the size of the initial $\+R$ is $O(\Delta |\+D|)$. The efficiency result in \Cref{theorem-dynamic} follows from the identical proof in \Cref{section-running-time}. 

\section*{Acknowledgement}
The research was supported by the National Key R\&D Program of China 2018YFB1003202 and NSFC under Grant Nos. 61722207 and 61672275. We thank Jingcheng Liu for helpful comments. Heng Guo and Yitong Yin want to thank the hospitality of the Simons Institute for the Theory of Computing at UC Berkeley, where part of the work was done.

\bibliographystyle{plain}
\bibliography{refs.bib}


\end{document}